\documentclass[10pt,a4paper,twocolumn]{article}
\usepackage[margin=1.5cm,top=2cm,bottom=2cm]{geometry}
\usepackage[T1]{fontenc}
\usepackage[pdfdisplaydoctitle,menucolor=orange!40!black,filecolor=magenta!40!black,urlcolor=blue!40!black,linkcolor=red!40!black,citecolor=green!40!black,ocgcolorlinks]{hyperref} %

\usepackage[utf8]{inputenc}
\usepackage[small]{caption}
\usepackage{graphicx,scalerel}
\usepackage{amsmath,amsthm,amssymb}
\usepackage{booktabs}
\usepackage[ruled,vlined,linesnumbered]{algorithm2e}

\newcommand{\ExternalLink}{%
    \tikz[color=magenta, x=1.2ex, y=1.2ex, baseline=-0.05ex]{%
        \begin{scope}[x=1ex, y=1ex]
            \clip (-0.1,-0.1) --++ (-0, 1.2) --++ (0.6, 0) --++ (0, -0.6) --++ (0.6, 0) --++ (0, -1);
            \path[draw, line width = 0.5, rounded corners=0.5] (0,0) rectangle (1,1);
        \end{scope}
        \path[draw, line width = 0.5] (0.5, 0.5) -- (1, 1);
        \path[draw, line width = 0.5] (0.6, 1) -- (1, 1) -- (1, 0.6);
    }
}

\DeclareUnicodeCharacter{1EF3}{\`y}
\usepackage[backend=biber,
  isbn=false,url=false,doi=true,
  backref=true,hyperref=true,natbib=true,
  maxcitenames=3,maxbibnames=99,
  sortcites]{biblatex}
\renewbibmacro{in:}{\ifentrytype{article}{}{\printtext{\bibstring{in}\intitlepunct}}}
\newbibmacro{doilink}[1]{%
  \iffieldundef{doi}{%
    \iffieldundef{url}{#1}{\href{\thefield{url}}{$\ExternalLink$}}}{\href{https://doi.org/\thefield{doi}}{$\ExternalLink$}}%
}
\DeclareFieldFormat*{doi}{\usebibmacro{doilink}{#1}}

\makeatletter
\def\@maketitle{%
\newpage\null\vskip 2em%
\begin{center}%
\let \footnote \thanks
  {\Large\bf \@title \par}\vskip 1.5em{\large\lineskip .5em\begin{tabular}[t]{c}\@author\end{tabular}\par}\vskip 1em{\large \@date}%
\end{center}\par\vskip 1.5em
}
\def\url@leostyle{\@ifundefined{selectfont}{\def\UrlFont{\sf}}{\def\UrlFont{\small\ttfamily}}}
\makeatother
\bibliography{../bibsy/strings-long,bicmce-bib}

\usepackage[dvipsnames,table]{xcolor}

\usepackage[mathic]{mathtools}
\usepackage{etoolbox}

\usepackage[capitalise]{cleveref}
\crefname{algocf}{Algorithm}{Algorithms}
\Crefname{algocf}{Algorithm}{Algorithms}
\crefname{axiom}{}{}
\creflabelformat{axiom}{(#2I#1#3)}

\usepackage{url}
\usepackage{dsfont, paralist, microtype} %
\usepackage{booktabs, tabularx}

\usepackage{tikz}
\usetikzlibrary{arrows, arrows.meta, calc, decorations.pathmorphing, decorations.pathreplacing, calligraphy, backgrounds, math, shapes, shapes.geometric, spy, positioning, quotes}

  \RequirePackage{xargs}     

  \usetikzlibrary{decorations.pathreplacing}
  \usetikzlibrary{math}
  \usetikzlibrary{calc}
  \usetikzlibrary{shapes}

  \newcommand{\tikzpreamble}{%
    \def\teps{0.075}
    \def\nsc{0.33}
    \def\tanS{1.7320508}
    \def\bsc{0.075}
    \def\colpNPh{red!85!black}
    \def\colXPWh{orange}
    \def\colFPTnoPK{green!80!black}
    \def\colFPTPK{green!40!white}
    \def\colFPT{green}
    \def\colXP{yellow}
    \def\colOpen{brown!20!white}
    \tikzstyle{xnode}=[circle,scale=\nsc,draw,fill=white];
    \tikzstyle{xnodeA}=[circle,scale=\nsc,fill=white,draw=red];
    \tikzstyle{xnodeB}=[circle,scale=\nsc,fill=lightgray,draw=green];
    \tikzstyle{xnodeC}=[circle,scale=\nsc,fill=black,draw=blue];
    \tikzstyle{xnodex}=[circle,fill,scale=\nsc,draw];
    \tikzstyle{xnodey}=[diamond,fill,scale=\nsc,draw];
    \tikzstyle{xstarL}=[star,star points=6,draw,scale=0.4];
    \tikzstyle{xstar}=[star,star points=8,draw,scale=0.4];
    \tikzstyle{xstarB}=[star,star points=10,draw,scale=0.4];
    \tikzstyle{xedge}=[thick,-];
    \tikzstyle{xedgex}=[thick,-,dashed];
    \tikzstyle{xedgedot}=[thick,-,dotted];
    
    \tikzstyle{xpath}=[color=blue,opacity=0.25,line cap=round,line width=6pt];
    \tikzstyle{xpathS}=[color=blue!40!white,opacity=0.3,line cap=round,line join=round,line width=5pt];
    \tikzstyle{xpathSx}=[color=green!40!black,opacity=0.3,line cap=round,line join=round,line width=5pt];
    \tikzstyle{xpathx}=[color=magenta,opacity=0.40,line cap=round,line width=6pt];
    \tikzstyle{xpathy}=[color=green,opacity=0.40,line cap=round,line width=6pt];
    
    \tikzstyle{xhili}=[circle,scale=1.25,opacity=0.25,fill,color=orange,draw];
    \tikzstyle{xhiliS}=[circle,scale=0.625,opacity=0.25,fill,color=orange,draw];
    \tikzstyle{xhiliIS}=[circle,scale=1.25,opacity=0.25,fill,color=magenta,draw];
    \tikzstyle{xxhiliS}=[rectangle,scale=0.85,opacity=0.25,fill,color=cyan,draw];
  }

  \newcommand{\Grid}[9]{
    \foreach\x in {0,...,#2}{
      \foreach \y in {0,...,#3}{
        \node (#1\x\y) at (#8*\x*\xr+#4*\xr,#9*\y*\yr+#5*\yr)[xnode,#6]{};
      }
    }
    \ifnum#3>0
      \pgfmathsetmacro\yx{int(#3 - 1)}
      \foreach \x in {0,...,#2}
        \foreach \y [count=\yi] in {0,...,\yx}  
          \draw[#7] (#1\x\y)--(#1\x\yi) ;
    \fi
    \ifnum#2>0
    \pgfmathsetmacro\yx{int(#2 - 1)}
    \foreach \x in {0,...,#3}
      \foreach \y [count=\yi] in {0,...,\yx}  
        \draw[#7] (#1\y\x)--(#1\yi\x) ;
    \fi
  }

\usepackage{xargs}
\newcommandx{\mydefenv}[4][3=A]{%
  \newtheorem{#1}{#2}
  \ifstrequal{#3}{A}{\crefname{#1}{#2}{#2s}}{\crefname{#1}{#2}{#3}}%
  \Crefname{#1}{#4{.}}{#4s{.}}
}

\theoremstyle{plain}
\mydefenv{theorem}{Theorem}{Thm}
\mydefenv{lemma}{Lemma}{Lem}
\mydefenv{problem}{Problem}{Prob}
\mydefenv{proposition}{Proposition}{Prop}
\mydefenv{conjecture}{Conjecture}{Conj}
\mydefenv{corollary}{Corollary}[Corollaries]{Cor}
\mydefenv{observation}{Observation}{Obs}
\theoremstyle{definition}
\mydefenv{definition}{Definition}{Def}
\mydefenv{construction}{Construction}{Constr}
\mydefenv{rrule}{Data Reduction Rule}{DRR}
\theoremstyle{remark}
\mydefenv{example}{Example}{Ex}
\mydefenv{remark}{Remark}{Rem}

\newcommand{\cqed}{\hfill$\diamond$}
\newcommand{\rqed}{\hfill$\triangleleft$}

\newcommand{\prob}[1]{\textnormal{\textsc{#1}}}
\newcommand{\wpb}{when parameterized by}
\newcommand{\RD}{$(\Rightarrow)\quad$}
\newcommand{\LD}{$(\Leftarrow)\quad$}
\newcommandx{\set}[2][1=1]{\ensuremath{\{#1,\ldots,#2\}}}
\newcommandx{\tlog}[3][1=,3=]{\log_{#1}^{#3}(#2)}
\newcommandx{\ith}[2][1=th]{#2\nobreakdash-#1}
\newcommandx{\decprob}[6][3=Input,5=Question]{\vspace{0.125em}
\begin{samepage}
\begingroup
\label{DEACTIVATEDprob:#2}{
{\noindent \textsc{#1}}}
  \nopagebreak[4]\nopagebreak[4]
  \par\noindent\hangindent=\parindent\textbf{#3}:  #4\nopagebreak[4]
  \par\noindent\hangindent=\parindent\textbf{#5}:  #6
  \par\medskip\endgroup
  \end{samepage}
}

\newcommand{\N}{\mathds{N}}
\newcommand{\Nzero}{\mathds{N}_0}

\newcommand{\R}{\mathds{R}}

\newcommand{\calC}{\mathcal{C}}

\newcommand{\calX}{\mathcal{X}}

\newcommand{\cocl}[1]{\textrm{#1}}
\newcommand{\WK}[1]{\cocl{WK[#1]}}
\newcommand{\WKone}{\WK{1}}
\newcommand{\FPT}{\cocl{FPT}}
\newcommand{\fpt}{fixed-parameter tractable}
\newcommand{\NP}{\cocl{NP}}
\newcommand{\classP}{\cocl{P}}
\newcommand{\coNP}{\cocl{coNP}}
\newcommand{\cpoly}{\cocl{poly}}
\newcommand{\NPincoNPslashpoly}{\ensuremath{\NP\subseteq\coNP/\cpoly}}
\newcommand{\unlessPK}{unless \NPincoNPslashpoly}
\newcommand{\UnlessPK}{Unless \NPincoNPslashpoly}
\newcommand{\Hbreaks}[1]{the #1 breaks}
\newcommand{\ETHbreaks}{\Hbreaks{ETH}}
\newcommand{\croco}{cross-composition}
\newcommand{\ORcroco}{OR-\croco}

\newcommand{\yes}{\emph{yes}}
\newcommand{\no}{\emph{no}}
\newcommand{\ceq}{\coloneqq}
\newcommand{\false}{\bot}
\newcommand{\true}{\top}

\newcommand{\tref}[1]{\scriptsize{(\Cref{#1})}}

\DeclareMathOperator{\poly}{poly}

\newcommand{\eps}{\varepsilon}

\newcommand{\Wilog}{Without loss of generality}

\newcommand{\ol}[1]{\overline{#1}}
\newcommand{\lneg}[1]{\ensuremath{\ol{#1}}}
\newcommand{\setto}{\leftarrow}

\newcommandx{\bicmce}[4][1=$\leq$,2=$\geq$,3=$\geq$]{(#1$\,\mid\,$#2,#3)\nobreakdash-\prob{BiCMCE$[$#4$]$}}
\newcommand{\smlaTsc}{\prob{Egalitarian Committee Sequence Election}}
\newcommand{\smlaAcr}{\prob{GCSE}}
\newcommand{\fmlaTsc}{\prob{Equitable Committee Sequence Election}}
\newcommand{\fmlaAcr}{\prob{QCSE}}
\newcommand{\pesmlaTsc}{\prob{Pre-Elected \smlaAcr}}
\newcommand{\pesmlaAcr}{\prob{PE-\smlaAcr}}
\newcommand{\pefmlaTsc}{\prob{Pre-Elected \fmlaAcr}}
\newcommand{\pefmlaAcr}{\prob{PE-\fmlaAcr}}
\newcommand{\xOTsatTsc}{\prob{Exactly 1-in-3 SAT}}
\newcommand{\xOTsatAcr}{\prob{X1-3SAT}}

\newcommand{\appsymb}{$\star$}
\newcommand{\appref}[1]{{\hyperref[proof:#1]{\appsymb}}}
\newcommand{\apprefX}[1]{{\hyperref[#1]{\appsymb}}}

\newcommand{\appendixsection}[1]{%
  \gappto{\appendixProofText}{\section{Additional Material for Section~\ref{#1}}\label{app:#1}}
}

\newcommand{\toappendix}[1]{%
  \gappto{\appendixProofText}{{#1}}
}
\newcommand{\appendixproof}[2]{%
  \gappto{\appendixProofText}{\subsection{Proof \texorpdfstring{of \cref{#1}}{}}\label{proof:#1}#2}
}

\title{Algorithmics of Egalitarian versus Equitable Sequences of Committees}

\usepackage{authblk}

\newcommand{\AFFA}{AFFA (BR~5207/1 and NI~369/15)}

\author[1]{Eva Michelle Deltl}
\author[1,2]{Till~Fluschnik\thanks{Supported by the DFG, project \AFFA.}$^,$}
\author[2]{Robert Bredereck}

\affil[1]{\small
  Technische Universität Berlin, Faculty~IV, Algorithmics and Computational Complexity, Germany\\
  \texttt{e.deltl@campus.tu-berlin.de}
}
\affil[2]{\small
  Institut für Informatik, TU Clausthal, Germany\\
  \texttt{robert.bredereck@tu-clausthal.de,till.fluschnik@tu-clausthal.de}
}
\date{}

\begin{document}

\maketitle

\begin{abstract}
We study the election of sequences of committees, where in each of $\tau$~levels (e.g.\ modeling points in time)
a committee consisting of $k$~candidates from a common set of $m$~candidates is selected.
For each level, each of $n$~agents (voters) may nominate one candidate whose selection would satisfy her.
We are interested in committees which are good with respect to the satisfaction per day and per agent.
More precisely, we look for \emph{egalitarian} or \emph{equitable} committee sequences.
While both guarantee that at least $x$~agents per day are satisfied,
egalitarian committee sequences ensure that each agent is satisfied in at least $y$~levels
while equitable committee sequences ensure that each agent is satisfied in exactly $y$~levels.
We analyze the parameterized complexity of finding such committees for the
parameters~$n,m,k,\tau,x$, and~$y$,
as well as combinations thereof.
\end{abstract}

\section{Prologue}
Consider the very basic committee selection scenario where every agent
may nominate one candidate for the committee.
The only committee that gives \emph{certain satisfaction to each} agent,
which we call \emph{egalitarian committee},
consist of all nominated candidates.
A committee that gives \emph{each agent the same satisfaction},
which we call \emph{equitable committee},
would also have to consist of all nominated candidates,
or of no candidate at all.
Either outcome appears impractical.
So,
aiming for an equitable or egalitarian committee seems pointless in this setting.

With a small twist, however, it becomes a meaningful yet unstudied case:
what happens when the agents can nominate candidates in different levels,
or, 
to put differently,
for different points in time?
Are there non-trivial egalitarian or equitable committee sequences?
Can we simultaneously guarantee a certain minimum number of nominations in each level?
And if so,
what is the computational complexity we have to face when trying to find such a committee?

What probably appears abstract at first glace is indeed quite natural:
when selecting the menu for some event, each participant may nominate a food option
(with levels being courses),
when organizing a panel, each organizer may nominate a session topic
(with levels being days with different topic frames),
or when planning activities as sketched next.

{\newcommand{\ff}[1]{\cellcolor{blue!25}\bf #1}
\begin{figure}[t]
 \begin{tabular}{l c c}
     ~   &     1 &     2\\
   $a_1$ &     D &     R\\
   $a_2$ &     S &     M\\
   $a_3$ &     D &     H\\
   $a_4$ &     S &     T\\
   $a_5$ &     M &     H\\
   $a_6$ &     R &     M
 \end{tabular}
 ~~~~~
 \begin{tabular}{l c c}
     ~   &     1 &     2\\
   $a_1$ & \ff D &     R\\
   $a_2$ & \ff S & \ff M\\
   $a_3$ & \ff D & \ff H\\
   $a_4$ & \ff S &     T\\
   $a_5$ &     M & \ff H\\
   $a_6$ &     R & \ff M
 \end{tabular}
 ~~~~~
 \begin{tabular}{l c c}
     ~   &     1 &     2\\
   $a_1$ & \ff D &     R\\
   $a_2$ &     S & \ff M\\
   $a_3$ & \ff D &     H\\
   $a_4$ &     S & \ff T\\
   $a_5$ & \ff M &     H\\
   $a_6$ &     R & \ff M
 \end{tabular}
 \caption{Illustration to~\cref{ex:intro1}.
 Left: Preferences indicating the favorite activity of each agent for each day.
          Middle: An egalitarian committee sequence. Right: An equitable committee sequence.}
 \label{fig:ex1}
\end{figure}
\begin{example}
\label{ex:intro1}
We want to bring together six agents at some weekend trip.
Each one announces 
what they want to do on each day of the weekend.
They will only form a group 
if each of them is happy with at least one of the chosen activities
over all days.
Possible activities are:
dancing~(D),
hiking~(H),
museum~(M),
restaurant~(R),
sightseeing~(S),
and
theater~(T).
The agents' preferences are given in \cref{fig:ex1}.
Assume we can choose two activities per day.
To get an overall good satisfaction, we aim to ensure that a strict majority of agents is satisfied each day (in addition to requiring each agent being satisfied at least once).
To realize this, we must select~$\{D,S\}$ for day one and~$\{M,H\}$ for day two.
While this egalitarian committee sequence indeed maximizes satisfaction per day, the agents might not find this fair, because some are satisfied on two days while others are only satisfied once.
We can fix this by aiming to ensure that each agent is satisfied \emph{exactly once} and only a weak majority of agents is satisfied each day.
To realize this, we select~$\{D,M\}$ for day one and~$\{M,T\}$ for day two, which gives an equitable committee sequence.
\end{example}

More formally, we study the following two problems and analyze their (parameterized) complexity
with respect to the following parameters and their combinations:
number~$n$ of agents,
number~$m$ of candidates,
number~$\tau$ of levels (e.g., 
time points),
size~$k$ of each committee,
number~$x$ of nominations the selected committee shall receive in each level, and
number~$y$ of successful nominations each agent makes in total.\footnote{In \cref{ex:intro1},
we have $n=m=6$, $\tau=k=2$, $y=1$, as well as $x=4$ in the egalitarian and $x=3$ in the equitable case.}

\decprob{\smlaTsc\,(\smlaAcr{})\!\nolinebreak}{smla}
{A set~$A$ of $n$~agents,
a set~$C$ of $m$~candidates,
a sequence of nomination profiles~$U=(u_1,\dots,u_\tau)$ with~$u_t\colon A\to C\cup\{\emptyset\}$,
and three integers~$k,x,y\in\Nzero$.}
{Is there a sequence~$C_1,\dots,C_\tau$ of subsets of~$C$ each of size at most~$k$ such that 
\begin{align}
 \text{$\forall t\in\set{\tau}$:} && |\{a\in A\mid u_t(a)\in C_t\}| &\geq x, \label{prob:smla:x}\\
 \text{and $\forall a\in A$:} && \sum\nolimits_{t=1}^\tau |u_t(a)\cap C_t| &\geq y? \label{prob:smla:y}
\end{align}
}

\noindent
We also refer to the left-hand side of~\eqref{prob:smla:x} 
and of~\eqref{prob:smla:y}
as
committee and agent \emph{score},
respectively.
\fmlaTsc{} (\fmlaAcr{}) denotes the variant where we replace~``$\geq$'' with~``$=$'' in \eqref{prob:smla:y}.

\paragraph{Related Work.}
From the motivations perspective, our model aims to select committees,
which is a well-studied core topic of computational social choice.
The three main goals of selecting committees discussed in the literature are individual excellence, proportionality, and diversity (cf.\ \citet{EFSS17}).
The latter is usually reached by egalitarian approaches~\cite{AFGST18} (on which we also focus), where %
the quality of a committee is defined by the least satisfied voter.

Our model considers preferences with more than one level.
Related,
in the 
\emph{multistage} setting~\cite{EMS14,GuptaTW14}
one finds committee election problems 
with multiple preferences for each agent~\cite{kellerhals2021parameterized,BFK22}.
While they also require a minimum satisfaction in each time step,
they do not require a minimum satisfaction of agents.
Instead, they have explicit constraints on the differences between two successive committees.

Also other aspects of selecting multiple (sub)committees have been studied before.
\citet{BKN20} augment classic multiwinner elections with a time dimension,
also selecting a sequence of committees.
The crucial differences with our work is that they do not allow agents (voters)
to change their ballots over time.
While \citet{FZC17}, \citet{Lac20}, and \citet{PP13} allow this, they consider
online scenarios in contrary to our offline scenario.
Moreover, they mostly focus on single-winner decisions and
evaluate the quality of solutions quite differently.
\citet{BHPRT21} also consider an offline setting but aim
for justified representation, a fairness notion for groups
of individuals.

\paragraph{Our Contributions.}

\cref{fig:results} gives a results overview
from our parameterized analysis.
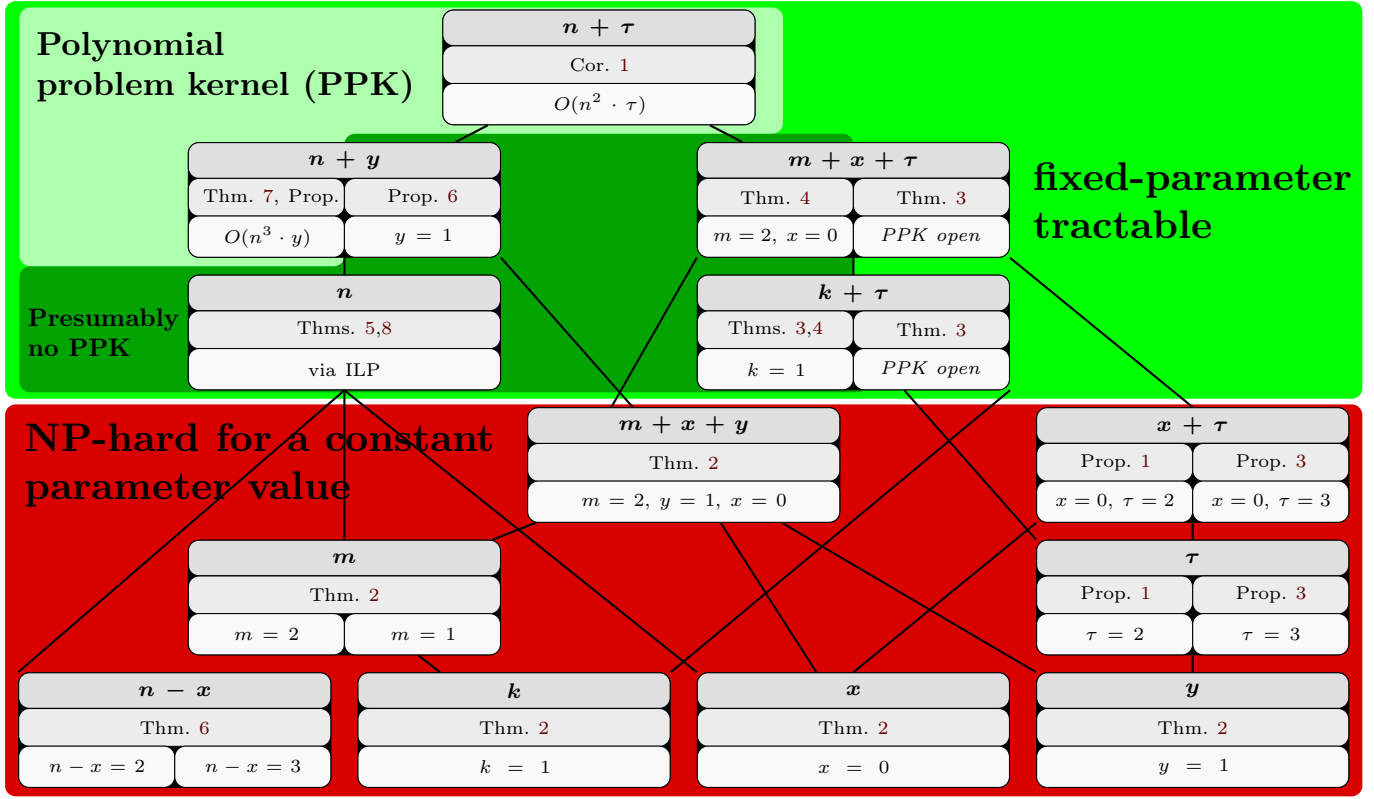
\begin{figure*}[t!]\centering
 \begin{tikzpicture}

  \def\xr{1.19}
  \def\yr{0.95}
  \def\bfntH{\footnotesize}
  \def\bfnt{\scriptsize}
  \tikzpreamble{}
  
  \def\Hcol{gray!25!white}
  \def\Mcol{gray!15!white}
  \def\Lcol{gray!5!white}
  
  \def\bw{3.45}
  \def\bh{1.6}
  \def\xsh{3.75}
  \def\ysh{1.85}
  \tikzstyle{Tparabox}=[anchor=north,minimum width=1*\bw*\xr cm,text width=0.8*\bw*\xr cm,align=center,minimum height=0.31*\bh*\yr cm,rounded corners,fill=\Hcol,draw]
  \tikzstyle{Lparabox}=[minimum width=1*\bw*\xr cm,text width=0.8*\bw*\xr cm,align=center,minimum height=0.36*\bh*\yr cm,rounded corners,fill=\Lcol,draw]
  \tikzstyle{subLparabox}=[minimum width=0.5*\bw*\xr cm,text width=0.4*\bw*\xr cm,align=center,minimum height=0.36*\bh*\yr cm,rounded corners,fill=\Lcol,draw]
  \tikzstyle{Mparabox}=[minimum width=1*\bw*\xr cm,text width=0.8*\bw*\xr cm,align=center,minimum height=0.33*\bh*\yr cm,rounded corners,fill=\Mcol,draw,yshift=1pt]
  \tikzstyle{subMparabox}=[minimum width=0.5*\bw*\xr cm,text width=0.4*\bw*\xr cm,align=center,minimum height=0.33*\bh*\yr cm,rounded corners,fill=\Mcol,draw,yshift=1pt]
  
  \newcommand{\paraboxA}[8]{%
    \node (n#1) at (#2*\xr,#3*\yr)[rounded corners, minimum width=1*\bw*\xr cm,minimum height=1*\bh*\yr cm,fill=black,draw]{};
    \node at (n#1.north)[Tparabox]{{\bfntH\boldmath#4}};
    \node at (n#1.west)[anchor=west,subMparabox]{\bfnt#5};
    \node at (n#1.east)[anchor=east,subMparabox,fill=\Mcol,draw]{\bfnt#6};
    \node at (n#1.south west)[anchor=south west,subLparabox]{\bfnt#7};
    \node at (n#1.south east)[anchor=south east,subLparabox]{\bfnt#8};
    
  }
  \newcommand{\paraboxAB}[7]{%
    \node (n#1) at (#2*\xr,#3*\yr)[rounded corners, minimum width=1*\bw*\xr cm,minimum height=1*\bh*\yr cm,fill=black,draw]{};
    \node at (n#1.north)[Tparabox]{{\bfntH\boldmath#4}};
    \node at (n#1.west)[anchor=west,subMparabox]{\bfnt#5};
    \node at (n#1.east)[anchor=east,subMparabox]{\bfnt#6};
    \node at (n#1.south)[anchor=south,Lparabox]{\bfnt#7};
    
  }
  \newcommand{\paraboxB}[6]{%
    \node (n#1) at (#2*\xr,#3*\yr)[rounded corners, minimum width=1*\bw*\xr cm,minimum height=1*\bh*\yr cm,fill=black,draw]{};
    \node at (n#1.north)[Tparabox]{{\bfntH\boldmath#4}};
    \node at (n#1.east)[anchor=east,Mparabox]{\bfnt#5};
    \node at (n#1.south)[anchor=south,Lparabox]{\bfnt#6};
  }
  \newcommand{\paraboxBA}[7]{%
    \node (n#1) at (#2*\xr,#3*\yr)[rounded corners, minimum width=1*\bw*\xr cm,minimum height=1*\bh*\yr cm,fill=black,draw]{};
    \node at (n#1.north)[Tparabox]{{\bfntH\boldmath#4}};
    \node at (n#1.east)[anchor=east,Mparabox]{\bfnt#5};
    \node at (n#1.south west)[anchor=south west,subLparabox]{\bfnt#6};
    \node at (n#1.south east)[anchor=south east,subLparabox]{\bfnt#7};
  }
  
  \newcommand{\thepboxes}{
    \paraboxB{k}{1*\xsh}{0}{$k$}{\Cref{thm:constmxyk}}{$k=1$}
    \paraboxB{x}{2*\xsh}{0}{$x$}{\Cref{thm:constmxyk}}{$x=0$}
    \paraboxB{y}{3*\xsh}{0}{$y$}{\Cref{thm:constmxyk}}{$y=1$}
    \paraboxBA{nmx}{0*\xsh}{0*\ysh}{$n-x$}{\Cref{thm:nmx}}{$n-x= 2$}{$n-x= 3$}
    \paraboxBA{m}{0.5*\xsh}{1*\ysh}{$m$}{\Cref{thm:constmxyk}}{$m=2$}{$m=1$}
    \paraboxA{tau}{3*\xsh}{1*\ysh}{$\tau$}{\Cref{thm:twolevels}}{\Cref{thm:threelevels}}{$\tau=2$}{$\tau=3$}
    \paraboxB{mpypx}{1.5*\xsh}{2*\ysh}{$m+x+y$}{\Cref{thm:constmxyk}}{$m=2$, $y=1$, $x=0$}
    \paraboxA{taupx}{3*\xsh}{2*\ysh}{$x+\tau$}{\Cref{thm:twolevels}}{\Cref{thm:threelevels}}{$x=0$, $\tau=2$}{$x=0$, $\tau=3$}
    \paraboxB{n}{0.5*\xsh}{3*\ysh}{$n$}{\Cref{thm:n,thm:eep}}{via ILP}
    \paraboxA{taupk}{2*\xsh}{3*\ysh}{$k+\tau$}{\Cref{thm:ktau,thm:nopkmtaux}}{\Cref{thm:ktau}}{$k=1$}{\emph{PPK open}}
    \paraboxA{npy}{0.5*\xsh}{4*\ysh}{$n+y$}{\Cref{thm:ny},~\Cref{thm:pknpy}}{\Cref{thm:FMLAnoPKnpy}}{$O(n^3\cdot y)$}{$y=1$}
    \paraboxB{nptau}{1.25*\xsh}{5*\ysh}{$n+\tau$}{\Cref{cor:pkntau}}{$O(n^2\cdot\tau)$}
    \paraboxA{mptaupx}{2*\xsh}{4*\ysh}{$m+x+\tau$}{\Cref{thm:nopkmtaux}}{\Cref{thm:ktau}}{$m=2$, $x=0$}{\emph{PPK open}}
  }
  \newcommand{\theparcs}{
    \foreach \x/\y in {
      m/n,
      k/m,
      x/mpypx,
      tau/taupx,
      y/tau,
      m/mpypx,
      taupk/mptaupx,
      mptaupx/nptau,
      n/npy,npy/nptau}{\draw[xedge] (n\x) to (n\y);}
    \draw[xedge] ($(nmpypx.north)+(-0.8*\xr,0)$) to (nmptaupx.south west);  
    \draw[xedge] (ntau.north west) to (ntaupk);
    \draw[xedge] (nx.north) to (ntaupx.south west);
    \draw[xedge] (ny.north west) to (nmpypx);
    \draw[xedge] (nk.north east) to (ntaupk.south east);
    \draw[xedge] (ntaupx.north) to (nmptaupx.south east);
    \draw[xedge] (nmpypx) to (nnpy.south east);
    \draw[xedge] (nnmx.north west) to (nn.south);
    \draw[xedge] (nx.north west) to (nn.south);
  }
  \def\tepsx{1.125*\teps}
  \def\tepsxx{2.25*\tepsx}
  \newcommand{\thecarea}{
    \fill[ultra thick,color=\colpNPh,rounded corners] 
      (-0.5*\xsh*\xr,2.5*\ysh*\yr-\teps) -- 
      (3.5*\xsh*\xr,2.5*\ysh*\yr-\teps) --
      (3.5*\xsh*\xr,-0.5*\ysh*\yr) --
      (-0.5*\xsh*\xr,-0.5*\ysh*\yr) -- cycle;
    \fill[ultra thick,color=\colFPT,rounded corners] 
      (-0.5*\xsh*\xr,2.5*\ysh*\yr) -- 
      (3.5*\xsh*\xr,2.5*\ysh*\yr) --
      (3.5*\xsh*\xr,5.5*\ysh*\yr) --
      (-0.5*\xsh*\xr,5.5*\ysh*\yr) -- cycle;
    \fill[ultra thick,color=\colFPTnoPK!80!black,rounded corners] 
      (2*\xsh*\xr+\tepsxx,2.5*\ysh*\yr+\tepsx) -- 
      (2*\xsh*\xr,3.5*\ysh*\yr) --
      (2*\xsh*\xr,4.5*\ysh*\yr) --
      (0.5*\xsh*\xr,4.5*\ysh*\yr) -- 
      (0.5*\xsh*\xr,3.5*\ysh*\yr) -- 
      (-0.5*\xsh*\xr+\tepsxx,3.5*\ysh*\yr) -- 
      (-0.5*\xsh*\xr+\tepsxx,2.5*\ysh*\yr+\tepsx) --cycle;
    \fill[ultra thick,color=\colFPTPK!80!white,rounded corners] 
      (-0.5*\xsh*\xr+\tepsxx,3.5*\ysh*\yr) -- 
      (0.5*\xsh*\xr,3.5*\ysh*\yr) --
      (0.5*\xsh*\xr,4.5*\ysh*\yr) --
      (1.75*\xsh*\xr+\tepsxx,4.5*\ysh*\yr) --
      (1.75*\xsh*\xr+\tepsxx,5.5*\ysh*\yr-\tepsx) --
      (-0.5*\xsh*\xr+\tepsxx,5.5*\ysh*\yr-\tepsx) -- cycle;
  }
  \thecarea{}
  \thepboxes{}
  \theparcs{}
  
  \node at (0.15*\xsh*\xr,5*\ysh*\yr)[scale=1.25,font=\bf, align=left]{Polynomial \\ problem kernel (PPK)};
  \node at (3*\xsh*\xr,4*\ysh*\yr)[scale=1.5,font=\bf, align=left]{fixed-parameter \\ tractable};
  \node at (-0.21*\xsh*\xr,3*\ysh*\yr)[scale=0.95,font=\bf, align=left]{Presumably \\ no PPK};
  \node at (0.25*\xsh*\xr,2*\ysh*\yr)[scale=1.5,font=\bf, align=left]{\NP-hard for a constant \\ parameter value};
 \end{tikzpicture}
 \caption{Overview of our results for~\smlaAcr{} and \fmlaAcr{}. 
 Each box has three horizontal layers,
 the top layer gives the parameter,
 the middle layer the result's reference,
 and the bottom layer gives additional information.
 If a box is vertically split,
 then the left and right side corresponds 
 to~\smlaAcr{} and~\fmlaAcr{},
 respectively;
 Otherwise,
 the information holds for both.
 The boxes are arranged according to the corresponding parameter hierarchy:
 If two boxes are connected by an edge,
 the upper one's parameter upper bounds the lower one's parameter (by some function).}
 \label{fig:results}
\end{figure*}
We highlight the following:
Each of \smlaAcr{} and \fmlaAcr{} is solvable in uniform polynomial time
\begin{itemize}
 \item for constantly many constant-size committees, 
 but not for constantly many committees 
 where each must have a committee score of at least a given constant (unless~$\classP=\NP$);
 or 
 \item for a constant number of agents, 
 but not for a constant number of candidates (unless~$\classP=\NP$).
\end{itemize}
We discovered the following differences between 
the egalitarian and equitable case:
\begin{itemize}
 \item For two stages, \smlaAcr{} is \NP-hard while
  \fmlaAcr{} is polyno\-mial-time solvable
  (\fmlaAcr{} is \NP-hard for three stages);
 \item For parameter~$n+y$, 
 \smlaAcr{} admits a polynomial problem kernel 
  while \fmlaAcr{} presumably does not;
 \item When~$k=m$,
  \smlaAcr{} is polynomial-time solvable,
  while~\fmlaAcr{} is still \NP-hard in this case.
  Notably,
  \smlaAcr{} is \NP-hard even if~$k=m-1$.
\end{itemize}
Due to the space constraints, many details, marked by~\appsymb, 
can be found in
the appendix.

\section{Preliminaries and Basic Observations}
\label{sec:basics}
\appendixsection{sec:basics}

We use standard notation from parameterized algorithmics~\cite{bluebook}.
A problem with parameter~$p$ is \fpt{} (in the class~$\FPT$),
if it can be solved in~$f(p)\cdot s^c$,
where~$s$ denotes the input size,
for some constant~$c$
and computational function~$f$ only depending on~$p$;
i.e.,
it can be solved in uniform polynomial time~$O(s^c)$ for every constant value of~$p$.
A (decidable) parameterized problem is \fpt{} if and only if it admits a problem kernel,
that is,
a polynomial-time algorithm that maps any instance with parameter~$p$
to an equivalent instance of size at most~$g(p)$,
where~$g$ is some function only depending on~$p$.
We speak of a polynomial problem kernel if~$g$ is a polynomial.

\toappendix{
We denote by~$B_N(i)$ the binary string of length~$2N$
with first~$N\in \N$ bits encoding number~$i-1\in\set[0]{2^N-1}$ on~$N$ bits, 
where the first bit is the leading bit, 
and the last~$N$ bits is the complement encoding;
e.g.,
$B_2(2)=(0,1,1,0)$.
For~$c\in C$,
we denote by~$u^{-1}_t(c)=\{a\in A\mid u_t(a)=c\}$.
}

\paragraph*{Basic Observations.}

We first discuss two trivial cases for \smlaAcr{} and \fmlaAcr{}
regarding the value of~$y$ 
and of~$k$.

\begin{observation}\label{obs:trivialy}
 If~$y\in\{0,\tau\}$,
 then~$\smlaAcr$ and~$\fmlaAcr$ are solvable in linear time.
\end{observation}

Note that \cref{obs:trivialy} implies that for~$\tau=1$,
each of \smlaAcr{} and \fmlaAcr{} is linear-time solvable.
Another trivial case for~\smlaAcr{} is the following.

\begin{observation}\label{obs:kgeqmSMLA}
 \smlaAcr{} is linear-time solvable if $k\geq m$.
\end{observation}

We will see that \cref{obs:kgeqmSMLA} does not transfer to~\fmlaAcr{}:
\fmlaAcr{} remains \NP-hard, even if~$k\geq m$ (\cref{thm:threelevels}). 

The following allows us to assume throughout 
to have at most number of agents many candidates.

\begin{lemma}[\appref{lem:agentmanycandidates}]
 \label{lem:agentmanycandidates}
 Each instance~$(A,C,U,k,x,y)$ of \smlaAcr{} (of \fmlaAcr{})
 can be mapped in linear time to an equivalent 
 instance~$(A,C',U',k,x,y), |C'|\leq |A|$ of \smlaAcr{} (of \fmlaAcr{}).
\end{lemma}

\appendixproof{lem:agentmanycandidates}
{
  \begin{proof}
  Let~$C'=\{c_1',\dots,c_n'\}$ be a set of candidates
  with~$C' \cap C = \emptyset$.
  Now, for each~$t\in\set{\tau}$,
  do the following 
  (consider~$A=\{a_1,\dots,a_n\}$ implicitly ordered).
  If~$i$ is the smallest index such that~$u_t(a_i)\in C$,
  and~$j$ the smallest index such that~$c_j'$ is not yet nominated,
  replace each nomination of~$u_t(a_i)$ by~$c_j'$.
  Correctness follows from the fact that~$u_t(a_i)=u_t(a_j) \iff u_t'(a_i)=u_t'(a_j)$ 
  for every~$i,j\in\set{n}$ and~$t\in\set{\tau}$.
  \end{proof}
}

\begin{corollary}
 \label{cor:pkntau}
 \begin{inparaenum}[(i)]
  \item Each of \smlaAcr{} and \fmlaAcr{}
    admits a problem kernel of size~$O(n^2\cdot\tau)$.
  \item There are at most~$(n+1)^n$ pairwise different nomination profiles.
 \end{inparaenum}
\end{corollary}

\section{Intractability}
\label{sec:intract}
\appendixsection{sec:intract}

We discuss the general intractability
of our problems as well as several special cases where
they remain hard.

\subsection{Dichotomies Regarding the Number of Levels}
\label{ssec:dichotau}

Both \smlaAcr{} and \fmlaAcr{} are easy problems if there is only one level.
Yet, 
already for two levels, 
\smlaAcr{} becomes \NP-hard while \fmlaAcr{} stays efficiently solvable.
For three levels,
however, 
also \fmlaAcr{} becomes \NP-hard.
We have the following.

\begin{theorem}
 \label{thm:tau}
 We have the following dichotomies for \smlaAcr{} and \fmlaAcr{} regarding~$\tau$:
 \begin{compactenum}[(i)]
  \item If~$\tau=1$, then
    each of \smlaAcr{} and \fmlaAcr{} is polynomial-time solvable.
  \item\label{thm:tau:ii} If~$\tau=2$, then
    \begin{inparaenum}[(a)]
     \item\label{thm:tau:ii:a} \smlaAcr{} is \NP-hard 
      and, \unlessPK,
      admits no problem kernel of size~$O(m^{2-\eps})$ for any~$\eps>0$,
      and 
     \item\label{thm:tau:ii:b} \fmlaAcr{} is polynomial-time solvable. 
    \end{inparaenum}
  \item\label{thm:tau:iii} If~$\tau\geq 3$, then
    each of \fmlaAcr{} with $k\geq m$ and \smlaAcr{} is \NP-hard
    and,
    unless \ETHbreaks,
    admits no $2^{o(n+m)}\cdot \poly(n+m)$-time algorithm.
 \end{compactenum}
\end{theorem}

We first discuss \eqref{thm:tau:ii:a},
then \eqref{thm:tau:ii:b},
and finally \eqref{thm:tau:iii}.

\subsubsection{Two Levels Make \smlaAcr{} Intractable}

\begin{proposition}[\appref{thm:twolevels}]
 \label{thm:twolevels}
 Even for two levels and~$x=0$,
 \smlaAcr{} is \NP-hard 
 and, \unlessPK,
 admits no problem kernel of size~$O(m^{2-\eps})$ for any~$\eps>0$.
\end{proposition}

The following problem is \NP-hard~\cite{KuoF87}.

\decprob{Constraint Bipartite Vertex Cover (CBVC)}{cbvc}
{An undirected bipartite graph~$G=(V,E)$ with~$V=V_1\uplus V_2$ and~$k_1,k_2\in\N$.}
{Is there a set~$X\subseteq V$ with~$|X\cap V_i|\leq k_i$ for each~$i\in\{1,2\}$ such that~$G-X$ contains no edge?}

Note that we can assume that~$k_1=k_2$.
CBVC is in \FPT{} \wpb{}~$k_1+k_2$~\cite{FernauN01}
but,
\unlessPK,
admits no problem kernel of 
size~$O(|V|^{2-\eps})$ for any~$\eps>0$~\cite{Jansen16}.
The construction behind the proof of~\cref{thm:twolevels}
is the following 
(the correctness proof is deferred to the appendix).

\begin{construction}\label{constr:twolevels}
 Let~$I=(G=(V=V_1\uplus V_2,E),k,k)$ be an instance of CBVC.
 We construct an instance~$I'\ceq (A,C,(u_1,u_2),k,x,y)$
 with~$x=0$ and~$y=1$ as follows.
 For each vertex~$v_{i,j}$ with~$i\in\{1,2\}$ and~$j\in\set{|V_i|}$,
 add a candidate~$c_{i,j}$ to~$C$.
 For each edge~$\{v_{1,j},v_{2,j'}\}$,
 add agent~$a_{j,j'}$ to~$A$
 which nominates~$c_{1,j}$ in level~$1$ and~$c_{2,j'}$ in level~$2$.
 This finishes the construction.
 \cqed
\end{construction}

\appendixproof{thm:twolevels}
{
  \begin{proof}[Proof of~\cref{thm:twolevels}]
    \RD{}
    Let~$X$ with~$X_i\ceq X\cap V_i$ be a solution to~$I$.
    We claim that~$C_i\ceq \{c_{i,j}\mid v_{i,j}\in X_i\}$,
    $i\in\{1,2\}$,
    is a solution to~$I'$.
    Note that~$|C_i|=|X_i|\leq k$ for each~$i\in\{1,2\}$.
    Moreover,
    suppose that there is an agent~$a_{j,j'}$ not being satisfied.
    Then,
    by construction,
    edge $\{v_{1,j},v_{2,j'}\}$ is not covered---a contradiction.
    
    \LD{}
    Let~$(C_1,C_2)$ be a solution to~$I$.
    We claim that $X=X_1\cup X_2$ with~$X_i\ceq \{v_{i,j}\mid c_{i,j}\in C_i\}$,
    $i\in\{1,2\}$,
    is a solution to~$I$.
    Note that~$|X_i|=|C_i|\leq k$ for each~$i\in\{1,2\}$.
    Moreover,
    suppose that there is an edge $\{v_{1,j},v_{2,j'}\}$ not being covered.
    Then,
    by construction,
    agent $a_{j,j'}$ is not satisfied---a contradiction.
  \end{proof}
}

\subsubsection{Two Levels Leave \fmlaAcr{} Tractable}

Interestingly,
in contrast to \smlaAcr{},
just one additional level does not change the tractability of \fmlaAcr.

\begin{proposition}
 \label{thm:tautwofmla}
 \fmlaAcr{} is polynomial-time solvable if~$\tau=2$.
\end{proposition}

We provide reduction rules for a generalization of \fmlaAcr{} on two levels,
and then reduce it to a special variant of CBVC.
The generalization of \fmlaAcr{} with~$\tau=2$ is the following.

\decprob{X2 \fmlaTsc{} (X2\fmlaAcr)}{X2smla}
{A set~$A$ of agents,
a set~$C$ of candidates,
a two nomination profiles~$U=(u_1,u_2)$ with~$u_t\colon A\to C\cup\{\emptyset\}$,
and five integers~$k_1,k_2,x_1,x_2,y\in\Nzero$.}
{Is there~$C_1\subseteq C$ with~$|C_1|\leq k_1$ and $C_2\subseteq C$ with~$|C_2|\leq k_2$ such that
\begin{align*}
 \text{$\forall t\in\{1,2\}$:} && |\{a\in A\mid u_t(a)\in C_t\}| &\geq x_t, \\
 \text{and $\forall a\in A$:} && \sum\nolimits_{t=1}^2 |u_t(a)\cap C_t| &= y?
 \end{align*}
}

\noindent
We know that~$y\in\{0,2\}$ are trivial cases.
Thus,
we assume that~$y=1$ is the remainder.
Our goal is to reduce X2\fmlaAcr{}
to the following problem,
which, as we will show subsequently,
is polynomial-time solvable.

\decprob{Constraint Bipartite Independent VC w/ Score (CBIVCS)}{cbivcs}
{An undirected bipartite graph~$G=(V,E)$ with~$V=V_1\uplus V_2$ and~$k_1,k_2,x_1,x_2\in\N$.}
{Is there an independent set~$X\subseteq V$ 
with~$|X\cap V_i|\leq k_i$ and~$\sum_{v\in X\cap V_i} \deg(v)\geq x_i$ for each~$i\in\{1,2\}$ 
such that~$G-X$ contains no edge?}

\begin{lemma}[\appref{prop:CBIVCSinP}]
 \label{prop:CBIVCSinP}
 \prob{CBIVCS} is polynomial-time solvable.
\end{lemma}

\appendixproof{prop:CBIVCSinP}
{
  The main ingredient is the following,
  which follows from the connectivity together with the unique 2-coloring.

  \begin{observation}
  \label{obs:ivcinccs}
  Let $I=(G=(V=V_1\uplus V_2,E),k_1,k_2,x_1,x_2)$ be a \yes-instance of \prob{CBIVCS} 
  and let $C_q=(V^q=V_1^q\cup V_2^q,E_q)$ be a connected component of~$G$.
  Then,
  for every solution~$X$ to~$I$,
  it holds true that either~$V_1^q\subseteq X$ or~$V_2^q\subseteq X$.
  \end{observation}

  \begin{proof}[Proof of \cref{prop:CBIVCSinP}]
  Let $(G=(V=V_1\uplus V_2,E),k_1,k_2,x_1,x_2)$ be an input instance of \prob{CBIVCS}.
  Let~$C_1=(V^1=V_1^1\cup V_2^1,E_1),\dots,C_p=(V^p=V_1^p\cup V_2^p,E_p)$ be the enumerated connected components of~$G$.
  We define~$N_j^i\ceq |V_j^i|$ and~$M_i\ceq |E_i|$.
  We define the following dynamic programming table:
  \begin{align}
    T&[i,k_1',k_2',x_1',x_2'] =\true \iff 
    \notag\\
    &\exists\text{ set~$X\subseteq \bigcup_{q=1}^i V^q$ with~$|X\cap \bigcup_{q=1}^i V^q_j|=k_j'$ and} \notag\\
      & \text{$\sum_{v\in X\cap \bigcup_{q=1}^i V^q_j} \deg(v)= x_j'$ for each~$j\in\{1,2\}$, and} \notag\\
      & \text{$X$ is an independent vertex cover of~$G[V^1\cup\dots\cup V^i]$.} \label{eq:cbivcs:dp:int}
  \end{align}
  Clearly, 
  $T[1,N_1^1,0,M_1,0]=T[1,0,N_2^1,0,M_1]\ceq \true$ and all other entries for~$i=1$ are set to~$\false$.
  Now we have
  \begin{align}
    T&[i,k_1',k_2',x_1',x_2']=\true \iff \notag\\
    & T[i-1,k_1'-N_1^i,k_2',x_1'-M_i,x_2']=\true \notag
    \\ &\lor T[i-1,k_1',k_2'-N_2^i,x_1',x_2'-M_i]=\true \label{eq:cbivcs:dp:rec}.
  \end{align}
  We return~$\yes$ if there is an entry~$T[p,k_1',k_2',x_1',x_2']$ with~$k_1'\leq k_1$, $k_2'\leq k_2$, $x_1'\geq x_1$, and $x_2'\geq x_2$,
  and~$\no$ otherwise.
  
  \smallskip\noindent\textit{Correctness:}
  We prove the statement via induction on the first entry~$i$, 
  using~\eqref{eq:cbivcs:dp:rec}.
  By construction,
  the statement holds true for~$i=1$.
  
  \RD{}
  Let~$T[i,k_1',k_2',x_1',x_2'] =\true$.
  As to~\eqref{eq:cbivcs:dp:rec},
  we have that either $T[i-1,k_1'-N_1^i,k_2',x_1'-M_i,x_2']=\true$ or $T[i-1,k_1',k_2'-N_2^i,x_1',x_2'-M_2]=\true$.
  Let,
  due to symmetry,
  $T[i-1,k_1'-N_1^i,k_2',x_1'-M_i,x_2']=\true$.
  By induction,
  we have that there is a set~$X\subseteq \bigcup_{q=1}^{i-1} V^q$ with~$|X\cap \bigcup_{q=1}^{i-1} V^q_1|=k_1'-N_1^i$
  and~$|X\cap \bigcup_{q=1}^{i-1} V^q_2|=k_2'$,
  as well as
  $\sum_{v\in X\cap \bigcup_{q=1}^{i-1} V^q_1} \deg(v)= x_1'-M_i$
  and $\sum_{v\in X\cap \bigcup_{q=1}^{i-1} V^q_2} \deg(v)= x_2'$.
  Then~$X'\ceq X\cup V_1^i$ fulfills the right side of~\eqref{eq:cbivcs:dp:int}.
  
  \LD{}
  Let~$X\subseteq \bigcup_{q=1}^i V^q$ with~$|X\cap \bigcup_{q=1}^i V^q_j|=k_j'$
  and $\sum_{v\in X\cap \bigcup_{q=1}^i V^q_j} \deg(v)= x_j'$ for each~$j\in\{1,2\}$.
  Note that~$X\cap V^i\in\{V_1^i,V_2^i\}$ since 
  $X$ is an independent vertex cover of~$G[V^1\cup\dots\cup V^i]$
  (see~\cref{obs:ivcinccs}).
  Let~$X\cap V^i = V_1^i$ and~$X'\ceq X\setminus V_1^i$.
  We have $X'\subseteq \bigcup_{q=1}^{i-1} V^q$ with~$|X'\cap \bigcup_{q=1}^{i-1} V^q_1|=k_1'-N_1^i$ 
  and $|X'\cap \bigcup_{q=1}^{i-1} V^q_2|=k_2'$
  and $\sum_{v\in X'\cap \bigcup_{q=1}^{i-1} V^q_1} \deg(v)= x_1'-M_i$
  and $\sum_{v\in X'\cap \bigcup_{q=1}^{i-1} V^q_2} \deg(v)= x_2'$.
  By induction,
  $T[i-1,k_1'-N_1^i,k_2',x_1'-M_i,x_2']=\true$,
  and hence,
  $T[i,k_1',k_2',x_1',x_2']=\true$.
  \end{proof}
}

To reduce X2\fmlaAcr{} to \prob{CBIVCS}
we have to deal with agents nominating none or only one candidate.
The first case is immediate.

\begin{rrule}
 \label{rrule:tautwofmla:trivno}
 If there is an agent nominating no candidate,
 then return \no.
\end{rrule}

If an agent nominates only one candidate in one level and none in the other,
we have to pick this nominated candidates.

\begin{rrule}[\appref{rrule:tautwofmla:onlyone}]
 \label{rrule:tautwofmla:onlyone}
 If there is an agent~$a^*$ nominating one candidate~$c^*$ in one level~$t\in\{1,2\}$,
 and none in the other level~$t'$,
 then do the following:
 Decrease~$k_t$ by one,
 $x_t$ by~$|\{a\in A\mid u_t(a)=c^*\}|$,
 replace each candidate in~$\{c'\in C\mid \exists a\in A:\: u_t(a)=c^* \land u_{t'}(a)=c'\}$ with~$\emptyset$,
 and delete all agents from $\{a\in A\mid u_t(a)=c^*\}$.
\end{rrule}

\appendixproof{rrule:tautwofmla:onlyone}
{
  \begin{proof}
  Let~$I=(A,C,U,k_1,k_2,x_1,x_2,y)$ be the input instance,
  and let the instance~$I'\ceq (A',C',U',k_1',k_2,x_1',x_2,y)$ be obtained by the reduction rule
  (we assume that the agent~$a^*$ nominates only one candidate~$c^*$ in the first level in~$I$).
  
  \RD{}
  Let~$(C_1,C_2)$ be a solution to~$I$.
  Note that due to~$a^*$,
  we have that~$c^*\in C_1$.
  Hence,
  no candidate in~$\{c'\in C\mid \exists a\in A:\: u_1(a)=c^* \land u_{2}(a)=c'\}$ is contained in~$C_2$.
  Hence,
  $(C_1',C_2')$ with~$C_1'\ceq C_1\setminus \{c^*\}$ and~$C_2'\ceq C_2$
  forms a solution to~$I'$.
  
  \LD{}
  Let~$(C_1',C_2')$ be a solution to~$I'$
  where~$C_2'\cap \{c'\in C\mid \exists a\in A:\: u_1(a)=c^* \land u_{2}(a)=c'\}=\emptyset$.
  Then~$(C_1,C_2)$ with~$C_1\ceq C_1'\cup\{c^*\}$ and~$C_2\ceq C_2'$
  is a solution to~$I$.
  \end{proof}
}

Using~\cref{rrule:tautwofmla:trivno} and~\ref{rrule:tautwofmla:onlyone}
exhaustively,
we can finally reduce X2\fmlaAcr{} to \prob{CBIVCS},
proving~\cref{thm:tautwofmla}.

\begin{observation}[\appref{lem:x2fmlaAcrToCBIVCS}]
 \label{lem:x2fmlaAcrToCBIVCS}
 There is a polynomial-time many-one reduction from \prob{X2\fmlaAcr} to \prob{CBIVCS}.
\end{observation}

\appendixproof{lem:x2fmlaAcrToCBIVCS}
{
  \begin{proof}
    Let~$(A,C,(u_1,u_2),k_1,k_2,x_1,x_2,y)$ be an instance of \prob{X2\fmlaAcr}.
    If~$y\in\{0,2\}$, then solve the problem in polynomial time and output a trivial \yes- or \no-instance of \prob{CBIVCS} accordingly.
    Assume~$y=1$.
    Construct the instance~$(G=(V=V_1\uplus V_2,E),k_1',k_2',x_1',x_2')$ as follows.
    Apply \cref{rrule:tautwofmla:onlyone} and~\ref{rrule:tautwofmla:trivno} exhaustively.
    If an application of~\cref{rrule:tautwofmla:trivno} returned no,
    then return a trivial no-instance.
    Otherwise,
    every agent nominates exactly two candidates.
    For each candidate~$c\in \bigcup_{a\in A} u_1(a)$,
    add a vertex~$v^1_c$ to~$V_1$,
    and for each candidate~$c\in \bigcup_{a\in A} u_2(a)$,
    add a vertex~$v^2_c$ to~$V_1$.
    Now,
    for each agent~$a\colon c\: c'$,
    add the edge~$\{v^1_c,v^2_{c'}\}$ to~$E$.
    This finishes the construction.
    
    The correctness follows from the one-to-one correspondence of agents and edges.
  \end{proof}
}

\subsubsection{Three Levels Make \fmlaAcr{} Intractable}

We have seen that \fmlaAcr{} is polynomial-time solvable if~$\tau\leq 2$.
This changes for~$\tau\geq 3$.

\begin{proposition}
 \label{thm:threelevels}
  For at least three levels and~$x=0$,
  each of \fmlaAcr{} with $k\geq m$ and \smlaAcr{} is \NP-hard
  and,
  unless \ETHbreaks,
  admits no $2^{o(n+m)}\cdot \poly(n+m)$-time algorithm.
\end{proposition}

For \smlaAcr{},
the proof is via a polynomial-time many-one reduction from the 
famous 
\NP-complete problem
\textsc{3-Satisfiability (3-SAT)},
which transfers the well-known ETH lower bound~\cite{ImpagliazzoPZ01}
as well as NP-hardness~\cite{GareyJ79}.
Given
a set~$X$ of~$N$ variables and a 3-CNF formula~$\phi=\bigwedge_{i=1}^M K_i$ over~$X$,
3-SAT asks whether
there is a truth assignment~$f\colon X \to \{\false,\true\}$ satisfying~$\phi$.

\begin{construction}\label{constr:threelevels}
 Let~$I=(X,\phi)$ be an instance of \prob{3-SAT} with $N$ variables and~$M$ clauses.
 We construct an instance~$I'\ceq (A,C,U,k,x,y)$ of \smlaAcr{}
 as follows
 (see \cref{fig:threelevels} for an illustration).
 \begin{figure}\centering
  \begin{tikzpicture}
    \def\xr{1}
    \def\yr{1}
    \def\ysh{0.4}
    \def\vfsize{0.75}
    \tikzpreamble{}
    
    \newcommand{\initlvls}{
      \node (t1) at (1*\xr,0.5*\yr)[]{$1$};
      \node (t2) at (2*\xr,0.5*\yr)[]{$2$};
      \node (t3) at (3*\xr,0.5*\yr)[]{$3$};
    }
      
    \newcommandx{\newagent}[5][2=]{%
      \node at (t1|-#1)[#2]{#3};
      \node at (t2|-#1)[#2]{#4};
      \node at (t3|-#1)[#2]{#5};
    }
    \newcommand{\theoutline}{
      \node (a00) at (0,-0*\ysh*\yr)[anchor=east,scale=\vfsize]{$\vdots$};
      \newagent{a00}[scale=\vfsize]{$\vdots$}{$\vdots$}{$\vdots$}
      \node (a1) at (0,-1*\ysh*\yr)[anchor=east]{$a_{i,1}$:};
      \newagent{a1}{$c_i$}{$\emptyset$}{$\ol{c_i}$}
      \node (a2) at (0,-2*\ysh*\yr)[anchor=east]{$a_{i,2}$:};
      \newagent{a2}{$\ol{c_i}$}{$c_i$}{$\emptyset$}
      \node (a3) at (0,-3*\ysh*\yr)[anchor=east]{$a_{i,3}$:};
      \newagent{a3}{$\emptyset$}{$\ol{c_i}$}{$c_i$}
      \node (a4) at (0,-4*\ysh*\yr)[anchor=east]{$\ol{a_{i,1}}$:};
      \newagent{a4}{$\ol{c_i}$}{$\emptyset$}{$c_i$}
      \node (a5) at (0,-5*\ysh*\yr)[anchor=east]{$\ol{a_{i,2}}$:};
      \newagent{a5}{$c_i$}{$\ol{c_i}$}{$\emptyset$}
      \node (a6) at (0,-6*\ysh*\yr)[anchor=east]{$\ol{a_{i,3}}$:};
      \newagent{a6}{$\emptyset$}{$c_i$}{$\ol{c_i}$}
      \node (a0) at (0,-7*\ysh*\yr)[anchor=east,scale=\vfsize]{$\vdots$};
      \newagent{a0}[scale=\vfsize]{$\vdots$}{$\vdots$}{$\vdots$}
      \draw [decorate,decoration={brace,amplitude=5pt},xshift=-5pt,yshift=0pt] (a6.south west) -- (a1.north west) node [black,midway,xshift=-10pt,anchor=east] {$A_i$};
    }
    
    \begin{scope}
      \initlvls{}
      \theoutline{}
      \node (c1) at (0,-8*\ysh*\yr)[anchor=east]{$a_r$:};
      \newagent{c1}{$c_i$}{$c_q$}{$\ol{c_p}$}
      \node (a0) at (0,-9*\ysh*\yr)[anchor=east,scale=\vfsize]{$\vdots$};
      \newagent{a0}[scale=\vfsize]{$\vdots$}{$\vdots$}{$\vdots$};
      
      \draw[thin,gray]($(a00.north)+(0.5*\xr,0)$) -- ($(a0.south)+(0.5*\xr,0)$);
      \draw[thin,gray]($(a00.north)+(0.5*\xr,0)$) -- ($(a00.north)+(3.5*\xr,0)$);
      \node at (5.5*\xr,0)[draw=none]{};
    \end{scope}

  \end{tikzpicture}
  \caption{Illustration to \cref{constr:threelevels}
    with~$K_r=(x_i\lor x_q \lor \ol{x_p})$.}
  \label{fig:threelevels}
 \end{figure}
 Let~$C\ceq \{c_i,\ol{c_i}\mid x_i\in X\}$.
 Let~$A_i\ceq \bigcup_{j=1}^3 \{a_{i,j},\ol{a_{i,j}}\}$ for each~$i\in\set{N}$.
 and~$A\ceq A_1\cup\dots\cup A_N\cup\{a_1,\dots,a_M\}$.
 See \cref{fig:threelevels} for the nominations.
 Let~$k\ceq N$,
 $x\ceq 0$,
 and~$y\ceq 1$.
 \cqed
\end{construction}

The construction provides the following key property 
when~$I'$ is a \yes-instance: for every variable,
exactly one of the two corresponding candidates must be in the committee.

\begin{lemma}[\appref{obs:threelevels}]
 \label{obs:threelevels}
 If~$I'$ is a \yes-instance,
 then for every solution~$(C_1,C_2,C_3)$ it holds true that~$|C_j\cap \{c_i,\ol{c_i}\}|=1$ and 
 $C_j\cap \{c_i,\ol{c_i}\} = C_{j'}\cap \{c_i,\ol{c_i}\}$
 for all~$j,j'\in\{1,2,3\}$ and~$i\in\set{N}$.
\end{lemma}

\appendixproof{obs:threelevels}
{
  \begin{proof}
  Suppose towards a contradiction
  that there is an~$i\in\set{N}$ and~$j\in\{1,2,3\}$ such that $|C_j\cap \{c_i,\ol{c_i}\}|=0$.
  \Wilog{} let~$j=1$.
  Then in one level~$j'$,
  it must be~$|C_{j'}\cap \{c_i,\ol{c_i}\}|=2$.
  \Wilog{} let~$j'=2$.
  Then,
  neither~$a_{i,1}$ nor~$\ol{a_{i,1}}$ are satisfied.
  Since $u_3(a_{i,1})\cup u_3(\ol{a_{i,1}})=\{c_i,\ol{c_i}\}$,
  it also holds true that~$|C_{3}\cap \{c_i,\ol{c_i}\}|=2$.
  It is immediate that three candidates must be picked over the three levels to satisfy each agent in~$A_{i'}$ for every~$i'\in\set{N}$.
  Recall that the total budget is~$3N$.
  Thus, we have that~$|C_1|+|C_2|+|C_3|> 3N$,
  contradicting the fact that~$(C_1,C_2,C_3)$ is a solution.
  Note that the case that there is an~$i\in\set{N}$ and~$j\in\{1,2,3\}$ such that $|C_j\cap \{c_i,\ol{c_i}\}|=2$
  implies that there is an~$i'\in\set{N}$ and~$j'\in\{1,2,3\}$ such that $|C_{j'}\cap \{c_{i'},\ol{c_{i'}}\}|=0$.
  
  So,
  assume that for all~$i\in\set{N}$ and~$j\in\{1,2,3\}$ we have that~$|C_j\cap \{c_i,\ol{c_i}\}|=1$.
  Suppose towards a contradiction that there are $j,j'\in\{1,2,3\}$ and~$i\in\set{N}$ such that $C_j\cap \{c_i,\ol{c_i}\} \neq C_{j'}\cap \{c_i,\ol{c_i}\}$.
  Let~$\ell\neq \ell'$ be such that~$a_{i,\ell}$ is satisfied by~$C_j$ and $a_{i,\ell'}$ is satisfied by~$C_{j'}$.
  Then, for~$j''\in\{1,2,3\}\setminus\{j,j'\}$,
  we have that~$u_{j''}(a_{i,\ell''})=\emptyset$ with $\ell''\in\{1,2,3\}\setminus\{\ell,\ell'\}$ by construction,
  yielding a contradiction.
  \end{proof}
}

\begin{proof}[Proof of \cref{thm:threelevels} (\smlaAcr)]
 \RD{}
 Let~$f$ be a satisfying truth assignment.
 We claim that~$(C',C',C')$ with
 $C'=\{c_i\in C \mid f(x_i)=\true\} \cup \{\ol{c_i}\in C \mid f(x_i)=\false \}$ is a solution to~$I'$.
 Clearly,~$|C'|=N$.
 Moreover,
 if~$f(x_i)=\true$,
 then~$a_{i,j}$ is satisfied in level~$j$,
 and~$\ol{a_{i,1}}$ is satisfied in level~$3$ and
 $\ol{a_{i,j}}$ with~$j\in\{2,3\}$ is satisfied in~level~$j-1$.
 If~$f(x_i)=\false$,
 then~$\ol{a_{i,j}}$ is satisfied in level~$j$,
 and~$a_{i,1}$ is satisfied in level~$3$ and
 $a_{i,j}$ with~$j\in\{2,3\}$ is satisfied in~level~$j-1$.
 Since~$f$ is satisfying,
 there is exactly one level~$t$ with~$a_r$ being satisfied.
 
 \LD{}
 Let~$(C_1,C_2,C_3)$ be a solution to~$I'$.
 From~\cref{obs:threelevels} we know that~$C'=C_1=C_2=C_3$ 
 and that~$C'\cap \{c_i,\ol{c_i}\}=1$ for all~$i\in\set{N}$.
 Let~$f(x_i)=\true$ if~$c_i\in C'$,
 and~$f(x_i)=\false$ otherwise.
 Clearly,
 $f$ is a truth assignment.
 Suppose it is not satisfying,
 i.e.,
 there is a clause~$K_r$ with no literal evaluated to true.
 Then,
 agent~$a_r$ is satisfied in no level,
 contradicting that~$(C_1,C_2,C_3)$ is a solution to~$I'$.
\end{proof}

For \fmlaAcr{},
yet using again~\cref{constr:threelevels},
we instead reduce from the \NP-hard problem
\xOTsatTsc{} (\xOTsatAcr) 
\cite{Schaefer78},
where,
given
a boolean 3-CNF formula~$\phi$ over a set~$X$ of variables,
the question is whether there is a truth assignment~$f\colon X \to \{\false,\true\}$ 
such that for every clause,
there is exactly one literal evaluated to true?

\noindent
Notably,
\cref{obs:threelevels} also holds true here.
In fact,
we can even allow~$k=2N$,
since for each variable only one candidate is chosen,
as otherwise there is an agent scoring more than once.
\subsection{Few Candidates Are of No Help}
\label{ssec:fewcandidates}

One could conjecture that it should be possible to guess the committees,
and hence get some, possibly non-uniformly polynomial running time
when the number of candidates is constant.
In this section,
we will show that this 
conjecture is wrong unless $\classP=\NP$:
each of \smlaAcr{} and \fmlaAcr{} are \NP-hard even for two candidates.

\begin{theorem}[\appref{thm:constmxyk}]
 \label{thm:constmxyk}
 Even for $x=0$,
 $k=1$,
 and~$y=1$,
 each of \smlaAcr{} with two candidates and 
 \fmlaAcr{} with one candidate is \NP-hard.
 Moreover, 
 unless the SETH breaks, 
 \smlaAcr{}
 admits no~$(2 -\eps)^\tau \cdot \poly(\tau + n)$-algorithm.
\end{theorem}

For \smlaAcr{},
we reduce
from the 
well-known \NP-complete problem
\textsc{Satisfiability (SAT)},
which transfers the well-known SETH lower bound~\cite{ImpagliazzoP01}
as well as NP-hardness~\cite{GareyJ79}.
Given
a set~$X$ of~$N$ variables and a CNF formula~$\phi=\bigwedge_{i=1}^M K_i$ over~$X$,
SAT asks whether 
there is a truth assignment~$f\colon X \to \{\false,\true\}$ satisfying~$\phi$.

The construction is quite intuitive:
Each level corresponds to a variable,
and each agent corresponds to a clause.
In each level,
if the corresponding variable appears as a literal in the agent's corresponding clause,
then the agent nominates a candidate regarding whether it appears negated or unnegated.

\begin{construction}
 \label{constr:constmxyk}
 Let~$I=(X,\phi)$ be an instance of~SAT.
 Construct an instance~$I'\ceq (A,C,U,k,x,y)$
 as follows.
 Let~$A\ceq\{a_1,\dots,a_M\}$,
 $C\ceq\{c_\true,c_\false\}$,
 $\tau\ceq N$,
 $x\ceq 0$, 
 $k\ceq 1$,
 and~$y\ceq 1$.
 In level~$i$,
 agent~$a_j$ nominates
 \[
  \begin{cases}
    c_\true, &\text{if $x_i$ appears unnegated in~$K_j$,}\\
    c_\false, &\text{if $x_i$ appears negated in~$K_j$, and}\\
    \emptyset, & \text{otherwise}.
  \end{cases}
 \]
 This finishes the construction.
 \cqed
\end{construction}

\appendixproof{thm:constmxyk}
{
  \begin{proof}
    Let~$I=(X,\phi)$ be an instance of~SAT.
    Construct instance~$I'\ceq (A,C,U,k,x,y)$
    using \cref{constr:constmxyk}.
    We prove that~$I$ is a \yes-instance if and only if $I'$~is a \yes-instance.
    
    \RD{}
    Let~$f\colon X\to\{\true,\false\}$ be a satisfying truth assignment.
    We claim that~$\calC\ceq (C_1,\dots,C_N)$
    with~$C_i\ceq \{c_{f(x_i)}\}$ is a solution to~$I'$.
    Suppose not,
    that is,
    there is an agent~$a_j$ being non-satisfied,
    i.e.,
    $\sum_{x_i\text{ appears as literal in }K_j} |u_i(a_j)\cap C_i|=0$.
    Thus,
    clause~$K_j$ is not satisfied by~$f$,
    a contradiction.
    
    \LD{}
    Let $\calC\ceq (C_1,\dots,C_N)$
    with~$C_i\neq \emptyset$ be a solution to~$I'$.
    We claim that~$f\colon X\to\{\true,\false\}$ with~$f(x_i)=\true$ if and only if~$c_\true\in C_i$ is a solution to~$I$.
    Suppose not,
    that is,
    there is a clause~$K_j$ such that no literal is evaluated to true.
    Then agent~$a_j$ is not satisfied in any level,
    contradicting the fact that $\calC$ is a solution to~$I'$.
    
    To get the statement for \fmlaAcr{},
    adapt as follows.
    Take~\xOTsatTsc{} as input problem and apply \cref{constr:constmxyk}.
    Then,
    add for each variable two new agents and another (new) level,
    where in both of the levels corresponding
    to the variable 
    one of the two agents nominates~$c_\top$ 
    and the other nominates~$c_\bot$;
    This enforces that for each variable a choice 
    (corresponding to true or false) 
    must be made.
  \end{proof}
}

\begin{remark}
 For \fmlaAcr{},
 we reduce from \xOTsatAcr{}
 (see previous section)
 where no variable appears negated~\cite{Schmidt2010d},
 where the construction is 
 very similar to~\cref{constr:constmxyk}
 (yet~$c_\false$ can be dropped).
 Hence,
 a lower bound based on the SETH
 as for \smlaAcr{} remains open for \fmlaAcr{}.
 \rqed
\end{remark}

\section{Tractability}
\label{sec:tract}
\appendixsection{sec:tract}

In this section,
we discuss non-trivial tractable cases of \smlaAcr{} and \fmlaAcr{}.
It turns out that fixed-parameter tractability starts with
the number~$n$ of agents or
the solution size~$k\cdot \tau$,
i.e.,
with the combination of the size~$k$ of each committee and the number~$\tau$ of levels.
As to the latter,
recall that each of \smlaAcr{} and \fmlaAcr{} is \NP-hard if either~$k$ is constant or~$\tau$ is constant.
Finally,
we discuss efficient and effective data reduction
regarding~$n$ and~$n+y$.

\subsection{Few Small Committees May be Tractable}

We show that each of \smlaAcr{} and \fmlaAcr{} 
when parameterized by the solution size~$k\cdot\tau$ 
is in \FPT{}.
That is,
we can deal with many agents and candidates,
as long as we are asked to elect few small committees.
We also show that \smlaAcr{} admits presumably
no problem kernel of size polynomial in~$m\cdot\tau$.

\begin{theorem}
  \label{thm:ktau}
 Each of \smlaAcr{} and \fmlaAcr{}
 is solvable in~$2^{k\cdot\tau^2}\cdot \poly(n+m+\tau)$ time,
 and hence \fpt{} \wpb{}~$k+\tau$.
\end{theorem}

\noindent
We introduce generalized versions of \smlaAcr{} and \fmlaAcr{}.
Intuitively, they allow to fix certain parts of the solution.
Moreover, one may request level-individual committee sizes,
level-individual numbers of nomination the committees shall receive,
and agent-individual numbers of successful nominations the agents
 still have to make.

\medskip
\decprob{\pesmlaTsc{} (\pesmlaAcr)}{pesmla}
{A set~$A$ of agents,
a set~$C$ of candidates,
a sequence of nomination profiles~$U=(u_1,\dots,u_\tau)$ with~$u_t\colon A\to C\cup\{\emptyset\}$,
integers~$x_t,k_t\in\Nzero$ for each~$t\in\set{\tau}$
and integers~$y_a\in\Nzero$ for each~$a\in A$.}
{Is there a sequence~$C_1,\dots,C_\tau\subseteq C$ with~$|C_t|\leq k_t$
for every~$t\in\set{\tau}$ 
such that
\begin{align}
 \text{$\forall t\in\set{\tau}$:} && \!\!|\{a\in A\mid u_t(a)\in C_t\}| &\geq x_t,
 \notag\\
 \text{and $\forall a\in A$:} && \sum\nolimits_{t=1}^\tau |u_t(a)\cap C_t| &\geq y_a. \label{prob:pesmla:y}
\end{align}
}

\noindent
\pefmlaTsc{} (\pefmlaAcr) denotes the variant 
when replacing~``$\geq$'' with~``$=$'' in \eqref{prob:pesmla:y}.

Each of \pesmlaAcr{} and \pefmlaAcr{} use slightly different approaches.
However,
the core idea is the same:
in any solution,
each agent has a fingerprint over all levels regarding whether or not
her candidate is elected into the respective committee.
Note that there are at most~$2^\tau$ fingerprints.
Hence,
we can guess such a fingerprint for any unsatisfied agent and branch.
Together with the fact that the sum of the committee sizes in the sequence is at most~$k\cdot \tau$,
the result follows.

Throughout,
we use the following.
Fix any agent~$a\in A$.
We define for~$A'\ceq  A\setminus\{a\}$ the utility function
\[u_t-u_t(a) \colon A'\to C\cup\{\emptyset\},\, (u_t-u_t(a))(a') \mapsto u_t(a')\setminus u_t(a).\]

\begin{algorithm}[t]
  \SetKwFunction{FMain}{main}
  \SetKwProg{Fn}{function}{:}{}
  \FMain{$(A,C,U,(k_t)_t,(x_t)_t,(y_a)_a)$}\;
  \Return{\no}\;
  \Fn{\FMain{$(A,C,U,(k_t)_t,(x_t)_t,(y_a)_a)$}}{
    \DontPrintSemicolon
    \lIf{$k_t<0$ for some~$t\in\set{\tau}$}{\textbf{break}}
    \If{$\forall a\in A:\: y_a \leq 0$}
    {
      \ForEach{$t\in\set{\tau}$}
        {
        \If{$x_t>0$}{\DontPrintSemicolon \lIf{any level-$t$ committee of $k_t$ most nominated candidates in~$C$ w.\,r.\,t.~$u_t$ scores less than~$x_t$}{\textbf{break}}}
        }
      \Return{\yes}
    }
    \DontPrintSemicolon
    \lIf{$\exists a\in A$ with~$y_a>0$ but no fingerpint with at least~$y_a$ non-empty entries}{\textbf{break}}
    Let~$a\in A$ be such that~$y_a>0$ with at least one fingerpint with at least~$y_a$ non-empty entries\;
    \ForEach(\tcp*[f]{$\leq 2^\tau$ many})
      {$X\in \{u_1(a),\emptyset\}\times\dots\times\{u_\tau(a),\emptyset\}$ with at least~$y_a$ non-empty entries}
      {

      \ForEach{$t\in\set{\tau}$}
      {
        Set~$x_t'\setto x_t-|\{a'\in A\mid u_t(a')=X_t\land X_t\neq \emptyset\}|$,
        $u_t'\setto u_t-u_t(a)$, and
        $k_t'\setto k_t-|X_t|$\;
      }
      \ForEach{$a'\in A'\setto A\setminus \{a\}$}{
        Set~$y_{a'}'\setto y_{a'}-\sum_t |u_t(a')\cap X_t|$\;
      }
      \FMain{$(A',C,U',(k_t')_t,(x_t')_t,(y_a')_a)$}\;
      }
  }
  \caption{FPT-algorithm for~\pesmlaAcr{} parameterized by~$k+\tau$ on input~$(A,C,U,(k_t)_t,(x_t)_t,(y_a)_a)$.}
  \label{alg:ktau:smla}
\end{algorithm}

We first show the following Turing-reduction for \pesmlaAcr{}.
The idea of this reduction is then used to obtain
fixed-parameter tractability through \cref{alg:ktau:smla}.

\begin{lemma}[\appref{lemma:turred:smla}]
 \label{lemma:turred:smla}
 Let $I\ceq (A,C,U,(k_t)_t,(x_t)_t,(y_a)_{a\in A})$ be an instance with
 at least one agent~$a\in A$ with~$y_a>0$ and 
 at least one fingerprint with at least~$y_a$ non-empty entries.
 Then,
 $I$ is a \yes-instance of \pesmlaAcr{}
 if and only if 
 for any agent $a\in A$ with~$y_a>0$ and 
 at least one fingerprint with at least~$y_a$ non-empty entries,
 one of the instances~$I^1,\dots,I^p$
 is a \yes-instance,
 where~$X^1,\dots,X^p\in \{u_1(a),\emptyset\}\times\dots\times\{u_\tau(a),\emptyset\}$ are the fingerprints with at least~$y_a$ non-empty entries 
 and for each~$q\in\set{p}$,
 $I^q=(A',C,U',(k_t^q)_t,(x_t^q)_t,(y_a^q)_{a\in A'})$,
 where~$A'\ceq A\setminus \{a\}$, 
 and 
 \begin{compactitem}
  \item for each $t\in\set{\tau}$, 
    $x_t'\ceq x_t-|\{a'\in A\mid u_t(a')=X_t^q\land X_t^q\neq \emptyset\}|$,
    $u_t'\ceq u_t-u_t(a)$, 
    $k_t'\ceq k_t-|X_t^q|$, 
    and 
  \item for each $a'\in A'$,
    $y_{a'}'\ceq y_{a'}-\sum_{t=1}^\tau |u_t(a')\cap X_t^q|$.
 \end{compactitem}
\end{lemma}

\appendixproof{lemma:turred:smla}
{
  \begin{proof}
  \RD{}
  Let~$\calC=(C_1,\dots,C_\tau)$ be a solution to~$I$.
  Let~$X^q\in \{u_1(a),\emptyset\}\times\dots\times\{u_\tau(a),\emptyset\}$ be the fingerprint
  of~$a$ regarding~$\calC$.
  We show that~$I^q$ is a \yes-instance
  by claiming that~$\calC'=(C_1',\dots,C_\tau')$ is a solution to~$I^q$,
  where~$C_t'\ceq C_t\setminus X_t^q$.
  Clearly,
  $|C_t'|=|C_t|-|X_t^q|\leq k_t-|X_t^q|=k_t'$.
  For~$a'\in A'$,
  we have~$\sum_{t=1}^\tau |u_t(a)\cap C_t'| = \sum_{t=1}^\tau (|u_t(a)\cap C_t| - |u_t(a)\cap X_t^q|) \geq y_a- \sum_{t=1}^\tau |u_t(a')\cap X_t^q| =y_a'$.
  Finally,
  for the~$x$-scores,
  note that
  $\sum_{c\in C_t'} |u^{-1}_t(c)| 
  = \sum_{c\in C_t} |u^{-1}_t(c)| - |\{a'\in A\mid u_t(a')=X_t^q\land X_t^q\neq \emptyset\}|
  \geq x_t- |\{a'\in A\mid u_t(a')=X_t^q\land X_t^q\neq \emptyset\}|
  = x_t'$.
  
  \LD{}
  Let~$\calC'=(C_1',\dots,C_\tau')$ be a solution to~$I^q$.
  We can assume that for every~$t\in\set{\tau}$ with~$u_t(a)\neq \emptyset$,
  $C_t'\cap u_t(a)=\emptyset$,
  since no agent nominates candidate~$u_t(a)$ by construction of~$u_t'$.
  By construction,
  we know that $a$'s fingerprint~$X^q$ contains at least $y_a$ non-empty entries.
  We claim that~$\calC\ceq (C_1,\dots,C_\tau)$ with~$C_t\ceq C_t'\cup X_t^q$
  is a solution to~$I$.
  Clearly,
  $|C_t|=|C_t'|+|X_t^q|\leq k_t'+|X_t^q|=k_t$.
  We know that the~$y$-score of~$a$ is fulfilled.
  For~$a'\in A'$,
  we have~$\sum_{t=1}^\tau |u_t(a)\cap C_t| = \sum_{t=1}^\tau (|u_t(a)\cap C_t'| + |u_t(a)\cap X_t^q|) \geq y_a'+ \sum_{t=1}^\tau |u_t(a')\cap X_t^q| =y_a$.
  Finally,
  for the~$x$-scores,
  note that
  $\sum_{c\in C_t} |u^{-1}_t(c)| 
  = \sum_{c\in C_t'} |u^{-1}_t(c)| + |\{a'\in A\mid u_t(a')=X_t^q\land X_t^q\neq \emptyset\}|
  \geq x_t'+ |\{a'\in A\mid u_t(a')=X_t^q\land X_t^q\neq \emptyset\}|
  = x_t$.
  \end{proof}
}

\begin{proposition}[\appref{prop:ktau:smla}]
 \label{prop:ktau:smla}
 \cref{alg:ktau:smla} is correct 
 and 
 runs in \FPT{}-time
 regarding~$k+\tau$.
\end{proposition}

\appendixproof{prop:ktau:smla}
{
  \begin{proof}
  Clearly,
  in each branch,
  we decrease~$k_t$ for at least one~$t\in\set{\tau}$.
  Thus,
  we have at most~$k\cdot\tau$ branches,
  where in each we check for at most~$2^\tau$ fingerprints.

  The correctness follows from~\cref{lemma:turred:smla}.
  The algorithms returns \yes{}
  if and only if 
  input instance
  $I=(A,C,U,(k_t)_t,(x_t)_t,(y_a)_a)$ is a \yes-instance.
  The forward direction is clear:
  if it returns true in some branch,
  then,
  due to~\cref{lemma:turred:smla},
  $I$ is a \yes-instance.
  Hence, 
  we prove the backward direction.
  
  We prove via induction on the number of agents.
  If $A=\emptyset$,
  then the algorithm returns \yes{}
  if and only if 
  $k_t\geq0$ and~$x_t\leq 0$ for all~$t$,
  which is correct.
  So let~$|A|\geq 1$.
  Let~$I$ be a \yes-instance.
  
  Since~$I$ is a \yes-instance,
  there is no~$t\in\set{\tau}$ with~$k_t<0$.
  If for all~$a\in A$ we have~$y_a\leq 0$,
  and there is a~$t\in\set{\tau}$ with~$x_t>0$,
  then there are at least~$k_t$ candidates available to make a score of at least~$x_t$.
  In this case,
  the algorithm returns~\yes{}.
  Finally,
  again since~$I$ is a \yes{}-instance,
  there is no~$a\in A$ with~$y_a>0$ 
  but no fingerpint with at least~$y_a$ non-empty entries.

  If the algorithm did not report \yes{} yet,
  it chooses a fingerprint~$X$ and recurses on instance~$I'=(A',C,U',(k_t')_t,(x_t')_t,(y_a')_a)$.
  Due to the \cref{lemma:turred:smla},
  we know that
  $I$ is a \yes-instance if and only if~$I'$ is a \yes-instance.
  Since~$I'$ has one agent less,
  by induction,
  the algorithm returns~\yes{}.
  \end{proof}
}

\noindent
The proof for \pefmlaAcr{} works very similarly 
and is hence deferred to the 
appendix.

In terms of kernelization,
we cannot improve much further:
Presumably,
there is no problem kernel of size polynomial in~$k+\tau$.
In fact,
we have the following stronger result.

\begin{theorem}[\appref{thm:nopkmtaux}]
 \label{thm:nopkmtaux}
 \UnlessPK,
 \smlaAcr{} admits no problem kernel of size polynomial in~$\tau$,
 even if~$m=2$ and~$x=0$.
\end{theorem}

\appendixproof{thm:nopkmtaux}
{
\begin{construction}
 \label{constr:nopkmtaux}
 Let~$I_1,\dots,I_p$ with~$p=2^q$, 
 $q\in\N$,
 be $p$-instances of \smlaAcr{} with 
 $I_i=(A_i=\{a_i^1,\dots,a_i^{|A_i|}\},
 C=\{c_0,c_1\},U_i=\{u_i^1,\dots,u_i^{\tau}\},k=1,x=0,y=1)$ for all~$i\in\set{p}$.
 Construct an instance~$I\ceq (A,C,U,k,x,y)$
 as follows.
 Let~$A=A_1\cup\dots\cup A_p$.
 Let~$U=(u_1,\dots,u_\tau,u_{\tau+1},\dots,u_{\tau+q})$
 where for~$i\in\set{\tau}$ we have~$u_i(a_j^{\ell})=u_j^i(a_j^{\ell})$
 and for~$i\in\set{q}$ we have
 \[ u_{\tau+i}(a_j^{\ell}) = 
    \begin{cases} 
    c_0,& \text{if there is a $0$ at the $i$th position} \\ &\text{of the $q$-binary encoding of~$j$}\\
    c_1,& \text{if there is a $1$ at the $i$th position} \\ &\text{of the $q$-binary encoding of~$j$}. 
    \end{cases}
 \]
 This finishes the construction.
 \cqed
\end{construction}

  \begin{proof}[Proof of~\cref{thm:nopkmtaux}]
  Let~$I_1,\dots,I_p$ with~$p=2^q$, 
  $q\in\N$,
  be $p$-instances of \smlaAcr{} with 
  $I_i=(A_i=\{a_i^1,\dots,a_i^{|A_i|}\},
  C=\{c_0,c_1\},U_i=\{u_i^1,\dots,u_i^{\tau}\},k=1,x=0,y=1)$ for all~$i\in\set{p}$.
  Construct instance~$I\ceq (A,C,U,k,x,y)$
  using~\cref{constr:nopkmtaux}.
  We prove that any of $I_1,\dots,I_p$ is a \yes-instance
  if and only if
  $I$ is a \yes-instance.
  
  \RD{}
  Let~$r\in\set{p}$ be such that~$I_r$ is a \yes-instance.
  Let~$(C_1,\dots,C_\tau)$ be a solution to~$I_r$.
  Let~$b_1\dots b_q$ be the $q$-binary encoding of~$r$.
  Let~$C_{\tau+i}=\{c_{b_i\oplus 1}\}$ for all~$i\in\set{q}$,
  where~$\oplus$ denotes the \emph{exclusive or}.
  We claim that~$\calC=(C_1,\dots,C_{\tau+q})$ is solution.
  Since~$(C_1,\dots,C_\tau)$ is a solution to~$I_r$,
  we know that each agent from~$A_r$ is satisfied within the first~$\tau$ layers.
  Now note that the $q$-binary encoding of every~$r'\neq r$ 
  has at least one bit in common with the encoding~$b_1\oplus 1 \dots b_q\oplus 1$.
  Thus,
  each agent of every instance different to~$I_r$ 
  is satisfied in one of the last~$q$ levels.
  Thus,
  $\calC$ is a solution to~$I$.
  
  \LD{}
  Let~$\calC=(C_1,\dots,C_{\tau+q})$ be a solution to~$I$
  such that~$C_i\neq \emptyset$ for all~$i\in\set{\tau+q}$.
  Let~$C_{\tau+i}=\{c_{b_i}\}$ for all~$i\in\set{q}$.
  Since~$p=2^q$,
  there is an instance~$I_r$ with $q$-binary encoding~$b_1\oplus 1\dots b_q\oplus 1$.
  Thus,
  no agent from~$A_r$ is satsified in any of the last~$q$ levels,
  and thus must be satisfied by~$\calC'\ceq (C_1,\dots,C_\tau)$.
  Hence,
  $\calC'$ is a solution to~$I_r$.
  \end{proof}
}

\begin{remark}\label{rem:mtaux}
 We leave open whether the composition can be adapted for \fmlaAcr{}.
 For this, the last~$q$ levels forming the selection gadget 
 must be changed or extended
 such that each agent gets the same score over the selection.
 \rqed
\end{remark}

\toappendix
{
  \subsubsection{\fmlaAcr{} is in FPT w.r.t.~\texorpdfstring{$k+\tau$}{k plus tau}}
  \label{proof:sssec:ktau:fmla}

  We first show the following Turing-reduction for \pefmlaAcr{}.

  \begin{lemma}
  \label{lemma:turred:fmla}
  Let $I\ceq (A,C,U,(k_t)_t,(x_t)_t,(y_a)_{a\in A})$ be an instance with
  at least one agent~$a\in A$ with~$y_a>0$ and 
  at least one fingerprint with exactly~$y_a$ non-empty entries.
  Then,
  $I$ is a \yes-instance of \pefmlaAcr{}
  if and only if 
  for any agent $a\in A$ with~$y_a>0$ and 
  at least one fingerprint with exactly~$y_a$ non-empty entries,
  one of the instances~$I^1,\dots,I^p$
  is a \yes-instance,
  where~$X^1,\dots,X^p\in \{u_1(a),\emptyset\}\times\dots\times\{u_\tau(a),\emptyset\}$ are the fingerprints with exactly~$y_a$ non-empty entries 
  and for each~$q\in\set{p}$,
  $I^q=(A',C,U',(k_t^q)_t,(x_t^q)_t,(y_a^q)_{a\in A'})$,
  where~$A'=A\setminus \{a\}$, 
  and 
  \begin{compactitem}
    \item for each $t\in\set{\tau}$, 
      $x_t'\ceq x_t-|\{a'\in A\mid u_t(a')=X_t^q\land X_t^q\neq \emptyset\}|$,
      $u_t'\ceq u_t-u_t(a)$, 
      $k_t'\ceq k_t-|X_t^q|$, 
      and 
    \item for each $a'\in A'$,
      $y_{a'}'\ceq y_{a'}-\sum_{t=1}^\tau |u_t(a')\cap X_t^q|$.
  \end{compactitem}
  \end{lemma}

  \begin{proof}
  \RD{}
  Let~$\calC=(C_1,\dots,C_\tau)$ be a solution to~$I$.
  Let~$X^q\in \{u_1(a),\emptyset\}\times\dots\times\{u_\tau(a),\emptyset\}$ be the fingerprint
  of~$a$ regarding~$\calC$.
  We show that~$I^q$ is a \yes-instance
  by claiming that~$\calC'=(C_1',\dots,C_\tau')$ is a solution to~$I^q$,
  where~$C_t'\ceq C_t\setminus X_t^q$.
  Clearly,
  $|C_t'|=|C_t|-|X_t^q|\leq k_t-|X_t^q|=k_t'$.
  For~$a'\in A'$,
  we have~$\sum_{t=1}^\tau |u_t(a)\cap C_t'| = \sum_{t=1}^\tau (|u_t(a)\cap C_t| - |u_t(a)\cap X_t^q|) = y_a- \sum_{t=1}^\tau |u_t(a')\cap X_t^q| =y_a'$.
  Finally,
  for the~$x$-scores,
  note that
  $\sum_{c\in C_t'} |u^{-1}_t(c)| 
  = \sum_{c\in C_t} |u^{-1}_t(c)| - |\{a'\in A\mid u_t(a')=X_t^q\land X_t^q\neq \emptyset\}|
  \geq x_t- |\{a'\in A\mid u_t(a')=X_t^q\land X_t^q\neq \emptyset\}|
  = x_t'$.
  
  \LD{}
  Let~$\calC'=(C_1',\dots,C_\tau')$ be a solution to~$I^q$.
  We can assume that for every~$t\in\set{\tau}$ with~$u_t(a)\neq \emptyset$,
  $C_t'\cap u_t(a)=\emptyset$,
  since no agent nominates candidate~$u_t(a)$ by construction of~$u_t'$.
  By construction,
  we know that $a$'s fingerprint~$X^q$ contains exactly~$y_a$ non-empty entries.
  We claim that~$\calC\ceq (C_1,\dots,C_\tau)$ with~$C_t\ceq C_t'\cup X_t^q$
  is a solution to~$I$.
  Clearly,
  $|C_t|=|C_t'|+|X_t^q|\leq k_t'+|X_t^q|=k_t$.
  We know that the~$y$-score of~$a$ is fulfilled.
  For~$a'\in A'$,
  we have~$\sum_{t=1}^\tau |u_t(a)\cap C_t| = \sum_{t=1}^\tau (|u_t(a)\cap C_t'| + |u_t(a)\cap X_t^q|) = y_a'+ \sum_{t=1}^\tau |u_t(a')\cap X_t^q| =y_a$.
  Finally,
  for the~$x$-scores,
  note that
  $\sum_{c\in C_t} |u^{-1}_t(c)| 
  = \sum_{c\in C_t'} |u^{-1}_t(c)| + |\{a'\in A\mid u_t(a')=X_t^q\land X_t^q\neq \emptyset\}|
  \geq x_t'+ |\{a'\in A\mid u_t(a')=X_t^q\land X_t^q\neq \emptyset\}|
  = x_t$.
  \end{proof}

  In contrast to~\pesmlaAcr{},
  if for an agent~$a$ we have~$y_a=0$,
  we know that no candidate of~$a$ can additionally be elected into any committee.
  Hence,
  we have the following.

  \begin{rrule}\label{rrule:fmlaktau}
  If for agent $a\in A$ we have that~$y_a=0$,
  then for all $t\in\set{\tau}$
  set $u_t\ceq u_t-u_t(a)$,
  and delete~$a$.
  \end{rrule}

  We point out that~\cref{rrule:fmlaktau}
  will asymptotically not speed up our algorithm;
  however,
  we see potential for improvement of any practical running time.

  \begin{algorithm}[t]
    \SetKwFunction{FMain}{main}
    \SetKwProg{Fn}{function}{:}{}
    \FMain{$(A,C,U,(k_t)_t,(x_t)_t,(y_a)_a)$}\;
    \Return{\no}\;
    \Fn{\FMain{$(A,C,U,(k_t)_t,(x_t)_t,(y_a)_a)$}}{
      \If{$k_t<0$ for some~$t$}{\textbf{break}}
      \If{$A\neq\emptyset$ and~$\exists a\in A:\: y_a<0$}{\textbf{break}}
      Apply \cref{rrule:fmlaktau} exhaustively\;
      \If{$A=\emptyset$ or~$\forall a\in A\colon y_a=0$}{
        \eIf{$\exists t\in\set{\tau}:\: x_t> 0$}{\textbf{break}}{\Return{\yes}}
      }
      \If{$\exists a\in A$ with~$y_a>0$ but no fingerpint with exactly~$y_a$ non-empty entries}{\textbf{break}}
      Let~$a\in A$ be such that~$y_a>0$ at least one fingerprint with exactly~$y_a$ %
      non-empty entries\;
      \ForEach(\tcp*[f]{$\leq 2^\tau$ many})
        {$X\in \{u_1(a),\emptyset\}\times\dots\times\{u_\tau(a),\emptyset\}$ with exactly $y_a$ %
         non-empty entries}
        {
        \ForEach{$t\in\set{\tau}$}
          {
            Set~$x_t'\setto x_t-|\{a'\in A\mid u_t(a')=X_t\land X_t\neq \emptyset\})$\;
            Set~$u_t'\setto u_t-u_t(a)$\;
            Set~$k_t'\setto k_t-|X_t|$\;
          }
          Set~$A'\setto A\setminus \{a\}$\;
          \ForEach{$a'\in A'$}{
            Set~$y_{a'}'\setto y_{a'}-\sum_t |u_t(a')\cap X_t|$\;
          }
        \FMain{$(A',C,U',(k_t')_t,(x_t')_t,(y_a')_a)$}\;
        }
    }%
    \caption{FPT-algorithm for~\pefmlaAcr{} parameterized by~$k+\tau$ on input~$(A,C,U,(k_t)_t,(x_t)_t,(y_a)_a)$.}
    \label{alg:ktau:fmla}
  \end{algorithm}

  \begin{proposition}
  \cref{alg:ktau:fmla} is correct 
  and 
  runs in \FPT{}-time
  regarding~$k+\tau$.
  \end{proposition}

  \begin{proof}
  Clearly,
  in each branch,
  we decrease~$k_t$ for at least one~$t\in\set{\tau}$.
  Thus,
  we have at most~$k\cdot\tau$ branches,
  where in each we check for at most~$2^\tau$ fingerprints.
  
  The correctness follows from~\cref{lemma:turred:fmla}.
  The algorithms returns \yes{}
  if and only if 
  input instance
  $I=(A,C,U,(k_t)_t,(x_t)_t,(y_a)_a)$ is a \yes-instance.
  The forward direction is clear:
  if it returns true in some branch,
  then,
  due to~\cref{lemma:turred:fmla},
  $I$ is a \yes-instance.
  Hence, 
  we prove the backward direction.
  
  We prove via induction on the number of agents.
  If $A=\emptyset$,
  then the algorithm returns \yes{}
  if and only if 
  $k_t\geq0$ and~$x_t\leq 0$ for all~$t$,
  which is correct.
  So let~$|A|\geq 1$.
  Let~$I$ be a \yes-instance.
  
  Since~$I$ is a \yes-instance,
  there is no~$t\in\set{\tau}$ with~$k_t<0$
  and no agent~$a\in A$ with~$y_a<0$.
  If for all~$a\in A$ we have~$y_a= 0$,
  then there is no~$t\in\set{\tau}$ with~$x_t>0$.
  In this case,
  the algorithm returns~\yes{}.
  Finally,
  again since~$I$ is \yes{},
  there is no~$a\in A$ with~$y_a>0$ 
  but no fingerpint with exactly~$y_a$ non-empty entries.

  If the algorithm did not report \yes{} yet,
  it chooses a fingerprint~$X$ and recurses on instance~$I'=(A',C,U',(k_t')_t,(x_t')_t,(y_a')_a)$.
  Due to the \cref{lemma:turred:fmla},
  we know that
  $I$ is a \yes-instance if and only if~$I'$ is a \yes-instance.
  Since~$I'$ has one agent less,
  by induction,
  the algorithm returns~\yes{}.
  \end{proof}
}

\subsection{Tractability Borders Regarding~\texorpdfstring{$n$}{n}}

We first show that both problems become fixed-parameter tractable when
parameterized by the number~$n$ of agents.

\begin{theorem}
  \label{thm:n}
 Each of \smlaAcr{} and~\fmlaAcr{} is \fpt{} \wpb{}~$n$.
\end{theorem}

\begin{proof}
 Due to~\cref{lem:agentmanycandidates},
 we know that there are at most~$n$ candidates,
 and 
 at most~$\nu\ceq (n+1)^n$ pairwise different nomination profiles.
 That is,
 we have at most~$\nu$ types,
 each having at most~$\binom{n}{k}$ committees of size~$k$ and score of at least~$x$
 (we call such a committee \emph{valid} subsequently; 
 note that we can check whether a committee is valid in linear time).
 
 Let~$x_{t,\phi}$ denote the variable for type~$t$ and valid committee~$\phi$. 
 Let~$n_t$ denote the number of type-$t$ profiles.
 For an agent~$a\in A$,
 let~$\calX_a$ denote the set of tuples~$(t,\phi)$ where valid committee~$\phi$ respects~$a$'s nomination in level~$t$.
 We then have the following integer programming constraints for \smlaAcr{}:
 
 \begin{align}
   \forall a\in A \colon && 
   \sum\nolimits_{(t,\phi)\in \calX_a} x_{t,\phi} &\geq y \label{al:n:atleasty:smla} 
   \\ 
   \forall t \colon && \sum\nolimits_{\text{valid~}\phi} x_{t,\phi} &= n_t \nonumber\\
   \forall t,\text{valid~}\phi \colon && 0\leq x_{t,\phi} &\leq n_t \nonumber
 \end{align}
 
 As to \citet{Lenstra83}, 
 having~$2^{O(n\tlog{n})}$ variables and constraints,
 and numbers upper bounded by~$\tau$,
 the result follows.
 For \fmlaAcr{},
 we replace~``$\geq$'' with~``$=$'' in~\eqref{al:n:atleasty:smla}.
\end{proof}

\noindent
\cref{thm:n} is in fact tight in the following sense:
decreasing~$n$ by~$x$ gives a useless parameter (presumably).

\begin{theorem}[\appref{thm:nmx}]
 \label{thm:nmx}
 \smlaAcr{} is \NP-hard even if~$n-x= 2$ and~$m=3$,
 and \fmlaAcr{} is \NP-hard even if~$n-x=3$ and~$m=2$.
\end{theorem}

\noindent
The 
construction behind the proof of~\cref{thm:nmx}
is very similar to~\cref{constr:constmxyk}
but with no empty nominations
(we hence defer also the construction to the appendix). %

\appendixproof{thm:nmx}
{

  \begin{construction}
  \label{constr:nmx}
  Let~$(X,\phi)$ be an instance of \prob{3-SAT} 
  where each clause contains exactly three literals and 
  every variable appears as a literal exactly two times negated and exactly two times unnegated,
  which remains \NP-hard~\cite{DarmannD21}.
  Construct~$(A,C,U,k,x.y)$ as follows.
  Let~$A\ceq\{a_1,\dots,a_M\}$
  and~$C\ceq \{c_\true,c_\false,c^*\}$. 
  Let~$\tau\ceq N$,
  where each level~$i$ correspond to variable~$x_i$.
  In level~$i$,
  agent~$j$ nominates
  $c_\true$, if~$x_i$ appears unnegated in~$K_j$,
  $c_\false$, if~$x_i$ appears negated in~$K_j$,
  and~$c^*$ otherwise.
  Let~$k\ceq 2$, 
  $y\ceq \tau-2$,
  and~$x\ceq M-2$.
  \cqed
  \end{construction}

  \begin{proof}
    The proof is analogous to the proof of~\cref{thm:constmxyk},
    with the difference that~$c^*$ is in every committee 
    (since~$x\geq M-4$)
    and hence every clause-agent is satisfied in~$\tau-3=y-1$ levels by~$c^*$.
    
    For~\fmlaAcr{},
    we reduce from \xOTsatAcr{} where every clause is of size exactly three and every variable appears exactly three times and never negated~\cite[Theorem~29]{Schmidt2010d}.
    Note that we hence can drop~$c_\false$ from~$C$.
    With~$x\ceq M-3$ the correctness now follows.
  \end{proof}
}

The FPT-algorithm behind~\cref{thm:n}
is not running in single-exponential time.
Combining~$n$ with~$y$ gives 
single-exponential running time.

\begin{theorem}%
  \label{thm:ny}
 Each of \smlaAcr{} and~\fmlaAcr{}
 is solvable in~$O((y+1)^n\cdot 2^n\cdot n\cdot \tau)$ time.
\end{theorem}

{
  \begin{proof}
  We give the proof for \fmlaAcr{},
  and it is not hard to adapt it for \smlaAcr{}.
  We use dynamic programming,
  where table
  
  \begin{itemize}
    \item[{$D[t,\mathbf{y}]$}] is true if and only if there are committees
      $C_1,\dots,C_t$ each with committee size at most~$k$ and a score of at least~$x$ such that the score of each agent~$a_i$ at time~$t$ sums up to exactly~$y_i$, 
      where~$\mathbf{y}=(y_1,\dots,y_n)$.
  \end{itemize}
  
  Set~$D[t,\mathbf{y}]$,
  where $t>1$ and each entry of~$\mathbf{y}$ is at most~$y$,
  to true if and only if 
  there is a set-to-true~$D[t-1,\mathbf{y'}]$ 
  and a size-at-most~$k$ score-at-least~$x$ committee~$C'\subseteq C$ with respect to~$u_t$
  such that~$\mathbf{y'}+\vec{c}=\mathbf{y}$,
  where~$\vec{c}=(c_1,\dots,c_n)\in\{0,1\}^n$ with
    $c_i = 0 \iff u_t(a_i)\cap C'=\emptyset$
  is called the fingerprint of~$C'$ regarding level~$t$.
  Set
  \[ D[1,\vec{c}~] \ceq \begin{cases} \true, & \text{if there is a size-at-most~$k$} \\ & \text{score-at-least~$x$ committee~$C'\subseteq C$} \\ & \text{with fingerprint~$\vec{c}$ regarding level~$1$, and} \\
  \false,& \text{otherwise}.
  \end{cases} \]
  Return~\yes{}
  if the entry~$D[\tau,(y_1,\dots,y_n)]$ is set to true,
  where~$y_1=y_2=\dots=y_n=y$,
  and~\no{} otherwise.
  
  The running time of filling the table is clear:
  We have at most~$\tau\cdot (y+1)^n$ entries
  and at most~$2^n$ different committees per level.
  We defer the correctness proof to the appendix.
  \end{proof}
  \appendixproof{thm:ny}
  {
    \begin{proof}
    We next prove the equivalence by induction on~$t$.
    The base case~$t=1$ is clear.
    Now assume that the equivalence holds true for up to~$t-1$.
    
    \RD{}
    By construction,
    there is a set-to-true~$D[t-1,\mathbf{y'}]$
    with~$\mathbf{y'}=(y_1',\dots,y_n')$
    and a size-at-most~$k$ score-at-least~$x$ committee~$C'\subseteq C$ with respect to~$u_t$
    such that~$\mathbf{y'}+\vec{c}=\mathbf{y}$,
    where~$\vec{c}=(c_1,\dots,c_n)$ is the fingerprint of~$C'$.
    By induction,
      there are committees
      $C_1,\dots,C_{t-1}$ each with committee size at most~$k$ and a score of at least~$x$ such that the score of each agent~$a_i$ at time~$t$ sums up to exactly~$y_i'$.
      Thus,
      for the sequence $C_1,\dots,C_{t-1},C_t\ceq C'$ we have that
      each committee has size at most~$k$ and a score of at least~$x$ such that the score of each agent~$a_i$ at time~$t$ sums up to exactly~$y_i'+c_i$.

    \LD{}
    Let $C_1,\dots,C_t$ be a sequence of committees each with size at most~$k$ and a score of at least~$x$ 
    such that the score of each agent~$a_i$ at time~$t$ sums up to exactly~$y_i$,
    and let~$\mathbf{y}=(y_1,\dots,y_n)$.
    Let~$\vec{c}=(c_1,\dots,c_n)$ be the fingerprint of~$C_t$.
    By induction, 
    we know that~$D[t-1,\mathbf{y}-\vec{c}]$ is true.
    By construction,
    since~$C_t$ is a size-at-most~$k$ score-at-least~$x$ committee,
    we know that $D[t,\mathbf{y}]$ is also set to true.
  \end{proof}
  }
}

\subsection{Efficient and Effective Data Reduction Regarding~\texorpdfstring{$n$}{n} and~\texorpdfstring{$y$}{y}}
\label{ssec:eep}
\appendixsection{ssec:eep}

While \smlaAcr{}
admits 
a problem kernel of size polynomial in~$n+y$,
\fmlaAcr{} does not presumably.
Moreover,
for \smlaAcr{},
dropping~$y$ also leads to kernelization lower bounds.
We have the following.

\begin{theorem}[\apprefX{app:ssec:eep}]
  \label{thm:eep}
  \UnlessPK,
  \begin{inparaenum}[(i)]
   \item \smlaAcr{}
    admits no problem kernel of size polynomial in~$n$,
    even if~$m=2$ and~$k=1$, and
   \item \fmlaAcr{}
    admits no problem kernel of size polynomial in~$n$,
    even if~$m=2$,
    $k=1$, and~$y=1$.
   \item \smlaAcr{}
    admits 
    a problem kernel of size polynomial in~$n+y$.
  \end{inparaenum}
\end{theorem}

We only discuss~(iii) briefly
(refer to the 
appendix 
for the remaining details).

\toappendix
{
\begin{proposition}%
 \label{thm:smla:nopkn}
 \UnlessPK,
 \smlaAcr{}
 admits no problem kernel of size polynomial in~$n$,
 even if~$m=2$ and~$k=1$.
\end{proposition}

{
We give a polynomial-parameter transformation~\cite{BTY11} from

\decprob{Multicolored Clique (MCC)}{mcc}
{A $k$-partite graph~$G=(V=V^1\uplus\dots\uplus V^k,E)$.}
{Is there $C\subseteq V$ such that~$G[C]$ is a clique and~$|C\cap V^i|=1$ for all~$i\in\set{k}$?}

\citet{HermelinKSWW15} proved that MCC \wpb{} $k\tlog{|V|}$
is \WKone-hard 
and admits no polynomial problem kernel
\unlessPK.

\begin{construction}\label{constr:smla:nopkn}
 Let~$I=(G=(V=V_1\uplus\dots\uplus V_k,E))$ be an input instance of MCC
 with~$V_i=\{v_i^1,\dots,v_i^N\}$ for all~$i\in\set{k}$,
 where~$N\ceq |V_i|$.
 Moreover, suppose that for each~$i\neq i'$,
 the number of edges between the two colors classes equals~$M>N$.
 We construct an instance $(A,C,U,k,x,y)$ of \smlaAcr{}  as follows.
 Let~$C=\{c^*,c_\true\}$
 and~$A=A_V \cup A_E \cup \bigcup_{i,i'\in\set{k},i\neq i'} A_{i,i'}$
 where~
 \begin{align*}
    A_V &\ceq \{a_i\mid i\in\set{k}\}, \\
    A_E &\ceq \{ a_{\{i,i'\}}\mid i,i'\in\set{k},i\neq i'\}, 
    \text{ and}\\
    A_{i,i'} &\ceq \{a^\ell_{i,i'}\mid \ell\in \set{2\tlog{N}}\}
    \\ &  \qquad\text{ for each~$i,i'\in\set{k},i\neq i'$.}  
 \end{align*}
 Let~$x\ceq 1$,
 $k\ceq 1$,
 and~$y\ceq M-1$.
 \begin{compactenum}[(i)]
  \item For each vertex~$v_i^j$,
 there is a level~$L_i^j$ with $u(a_i)=c^*$ and 
 for all~$\ell\in\set{2\tlog{N}}$ and~$i'\in\set{k}\setminus\{i\}$,
 $a^\ell_{i,i'}$ nominates~$c_\true$ if there is a $0$ at position~$\ell$ of~$B_N(j)$.
  \item For each edge~$\{v_i^j,v_{i'}^{j'}\}$,
 there is a level~$L_{i,i'}^{j,j'}$ with $u(a_{\{i,i'\}})=c^*$ and 
 $a^\ell_{i,i'}$ nominates~$c_\true$ if there is a $1$ at position~$\ell$ of~$B_N(j)$ and
 $a^{\ell'}_{i',i}$ nominates~$c_\true$ if there is a $1$ at position~$\ell'$ of~$B_N(j')$.
  \item There are~$M-2$ levels~$L_1^\dagger,\dots,L_{M-2}^\dagger$ 
 where only the agents in~$A_{i,i'}$ for all~$i,i'\in\set{k}$ nominate~$c^*$,
 and all other agents nominate nothing.
  \item Finally,
 there are ~$M-N$ levels~$L_1^\ddagger,\dots,L_{M-N}^\ddagger$
 where only the agents in~$A_V$ nominate~$c^*$,
 and all other agents nominate nothing.
 \end{compactenum}
 This finishes the construction.
 \cqed
\end{construction}

\begin{observation}%
 \label{obs:smla:nopkn:atmost}
 Let~$I'$ be a \yes-instance.
 For every solution,
 the following holds true:
 \begin{compactenum}[(i)]
  \item For every~$i\in\set{k}$,
 there is at most one~$j\in\set{N}$
 such that the committee in level~$L_i^j$ contains~$c_\true$.
  \item For every pair~$i,i'\in\set{k}$,
 there is at most one pair~$j,j'\in\set{N}$
 such that the committee in level~$L_{i,i'}^{j.j'}$ 
 contains~$c_\true$.
 \end{compactenum}
\end{observation}

{
  \begin{proof}
  (i) Suppose not,
  i.e.,
  there is $i\in\set{k}$,
  with at least two distinct~$j,j'\in\set{N}$
  such that the committee in level~$L_i^j$ and~$L_i^{j'}$ contains~$c_\true$.
  Then,
  agent~$a_i$ gets score at most~$(M-N)+(N-2)<y$,
  a contradiction.
  
  (ii) Suppose not,
  i.e.,
  there is a pair~$i,i'\in\set{k}$,
  with at least two pairs~$j_1,j_1'\in\set{N}$ and $j_2,j_2'\in\set{N}$
  such that the committee in level~$L_{i,i'}^{j_1,j_1'}$
  and $L_{i,i'}^{j_2,j_2'}$ 
  contains~$c_\true$.
  Then,
  agent~$a_{\{i,i'\}}$ gets score at most~$M-2<y$,
  a contradiction.
  \end{proof}
}

\begin{observation}%
 \label{obs:smla:nopkn:exact}
 Let~$I'$ be a \yes-instance.
 For every solution,
 the following holds true:
 \begin{compactenum}[(i)]
  \item For every~$i\in\set{k}$,
 there is exactly one~$j\in\set{N}$
 such that the committee in level~$L_i^j$ contains~$c_\true$.
  \item For every~$i,i'\in\set{k}$,
 there is exactly one pair~$j,j'\in\set{N}$
 such that the committee in level~$L_{i,i'}^{j.j'}$ 
 contains~$c_\true$.
 \end{compactenum}
\end{observation}

{
  \begin{proof}
  (i) Suppose not,
  i.e.,
  there is none for color~$i$.
  Then to have agent~$a_{i,i'}^\ell$ for each~$\ell$
  be satisfied,
  there must be two levels~$L_{i,i'}^{\cdot,\cdot}$ containing~$c_\true$,
  a contradiction to~\cref{obs:smla:nopkn:atmost}.
  
  (i) Suppose not,
  i.e.,
  there is none for colors~$i,i'$.
  Then to have agent~$a_{i,i'}^\ell$ for each~$\ell$
  be satisfied,
  there must be two levels~$L_{i}^{\cdot}$ containing~$c_\true$,
  a contradiction to~\cref{obs:smla:nopkn:atmost}.
  \end{proof}
}

\begin{observation}%
 \label{obs:smla:nopkn:match}
 Let~$I'$ be a \yes-instance.
 For every solution,
 if level~$L_i^j$'s committee contains~$c_\true$,
 then for every~$i'\in\set{k}$ there is a~$j_{i'}\in\set{N}$ such that
 the committee in~$L_{i,i'}^{j,j_{i'}}$ contains~$c_\true$.
\end{observation}

{
  \begin{proof}
  Suppose not,
  i.e.,
  the level containing~$c_\true$ is~$L_{i,i'}^{j',j_{i'}}$. with~$j'\neq j$.
  Then there is an~$\ell\in\set{2\tlog{N}}$ such that
  $a_{i,i'}^\ell$ is not satisfied:
  $\ell$ is the position where~$B_N(j)$ has a 1 
  and~$B_N(j')$ has a 0
  (which exists since~$j\neq j'$).
  \end{proof}
}

\begin{proof}[Proof of \cref{thm:smla:nopkn}]
 \RD{}
 Let~$C=\{v_i^{j_i}\mid i\in\set{k}\}$ form a multicolored clique.
 We construct a solution as follows.
 In each level~$L_{i}^{j_i}$ and~$L_{i,i}^{j_i,j_{i'}}$,
 candidate~$c_\true$ is elected.
 In all other levels,
 candidate~$c^*$ is elected.
 Observe that each agent in~$A_V\cup A_E$ is satisfied
 (recall that there are~$M-N$ levels satisfying each agent from~$A_V$).
 Consider any agent~$a_{i,i'}^\ell$.
 Note that $a_{i,i'}^\ell$ is satisfied in each of the~$M-2$ levels~$L_1^\dagger,\dots,L_{M-2}^\dagger$.
 Since in levels $L_{i}^{j_i}$ and~$L_{i,i}^{j_i,j_{i'}}$
 candidate~$c_\true$ is elected,
 if the~$\ell$'s position in~$B_N(j_i)$ is a 0,
 then~$a_{i,i'}^\ell$ is satisfied in level $L_{i}^{j_i}$,
 otherwise,
 by construction,
 $a_{i,i'}^\ell$ is satisfied in level $L_{i,i}^{j_i,j_{i'}}$.
 
 \LD{}
 Follows from \cref{obs:smla:nopkn:atmost,obs:smla:nopkn:exact,obs:smla:nopkn:match}.
\end{proof}
}
}

\begin{proposition}
 \label{thm:pknpy}
 \smlaAcr{}
 admits 
 a problem kernel of size polynomial in~$n+y$.
\end{proposition}

\newcommand{\critical}{critical}
\newcommand{\ncritical}{non-critical}
\newcommand{\validS}{Z}

In the following,
we 
(again) 
call a committee \emph{valid} if its size is at most~$k$ and its score is at least~$x$.
For an agent~$a$,
we denote by~$\validS(a)$ the set of all levels
where there is a valid committee containing $a$'s nominated candidate.
We call an agent $a$ \emph{\ncritical} if $|\validS(a)|>n\cdot y$,
and \emph{\critical} otherwise.
We have the followings.

\begin{rrule}[\appref{rr:onlynoncritical}]
 \label{rr:onlynoncritical}
 If every agent~$a$ is \ncritical{},
 then return a trivial \yes-instance.
\end{rrule}

\appendixproof{rr:onlynoncritical}{
  \begin{proof}
  Via induction on the number~$n$ of agents.
  For~$n=1$ the statement is clear.
  Hence,
  let the statement hold true for~$n-1>1$.
  Having $n$ agents where each agent~$a$ has at least~$n\cdot y$ levels 
  where there is a valid committee containing $a$'s nominated candidate.
  For some arbitrary agent~$a$,
  choose exactly $y$ levels~$t_1,\dots,t_y$ and valid committees $C_{t_1},\dots,C_{t_y}$
  such that $u_{t_i}(a)\in C_{t_i}$ for every $i\in\set{y}$.
  Note that $a$ is satisfied.
  Delete~$a$ and the levels~$t_1,\dots,t_y$.
  Note that for the remaining $n-1$ agents,
  we have that every agent~$a$ has at least~$(n-1)\cdot y$ levels 
  where there is a valid committee containing $a$'s nominated candidate.
  Due the the inductive hypothesis,
  we can return \yes.
  \end{proof}
}

Thus,
if we have a non-trivial instance,
then
there must be a critical agent.
We will see that the number of critical agents 
can upper bound the number of levels.
To this end,
we first delete levels which are irrelevant to critical agents
as follows.

\begin{rrule}%
 \label{rr:onlynoncriticallevels}
 If there is a level $t^*$ such that 
 there is at least one valid committee and 
 every valid committee only includes 
 candidates nominated by \ncritical{} agents,
 then delete this level.
\end{rrule}

{
  \begin{proof}
  Let~$I=(A,C,U,k,x,y)$ be the input instance
  and~$I'\ceq (A,C,U',k,x,y)$ be the instance obtained by the reduction rule.
  Clearly,
  if $I'$ is a \yes-instance,
  then $I$ is a \yes-instance.
  Hence,
  we show the converse next.
  
  Assume towards a contradiction
  that for every solution $\calC\ceq (C_1,\dots,C_{t^*},\dots,C_\tau)$
  it holds true that $\calC'\ceq (C_1,\dots,C_{t^*-1},C_{t^*+1},\dots,C_\tau)$
  is no solution to~$I'$,
  i.e.,
  there is a maximal set~$A^*\subseteq A$ of agents 
  which is not satisfied when $C_{t^*}$ is dropped.
  Let~$q\ceq|A^*|$.
  Recall that~$A^*$ consists of only \ncritical{} agents.
  Let $\ol{T}\subseteq\set{\tau}\setminus\{t^*\}$ be a minimum-size set of levels such that 
  all agents of $A\setminus A^*$ are satisfied
  (which exists since $\calC'$ satisfies all agents except for those in~$A^*$).
  Note that~$|\ol{T}|\leq (n-q)\cdot y$.
  Let~$T\ceq \set{\tau}\setminus (\ol{T}\cup\{t^*\})$.
  Hence,
  for all~$a\in A^*$,
  we have that~$|\validS(a)\cap T|\geq |\validS(a)|-|\ol{T}\cup\{t^*\}|\geq n\cdot y +1 - ((n-q)\cdot y+1) \geq q\cdot y$.
  Thus,
  there is a solution to~$I'$,
  yielding a contradiction.
  \end{proof}
}

It follows that in every level,
there must be a valid committee for any of the at most~$n$ critical agents,
each of which has at most~$n\cdot y$ levels of this kind.
This leads to the following.

\begin{lemma}[\appref{lem:PKnyInapp}]
 \label{lem:PKnyInapp}
 If each of \cref{rr:onlynoncritical} and~\ref{rr:onlynoncriticallevels}
 is inapplicable,
 then there are at most~$n^2\cdot y$ levels.
\end{lemma}

\appendixproof{lem:PKnyInapp}
{
\begin{proof}
 Due to the inapplicability of \cref{rr:onlynoncritical}
 there is at least one critical agent.
 Due to the inapplicability of \cref{rr:onlynoncriticallevels},
 we have that
 $\set{\tau}=\bigcup_{a\text{ \critical}} \validS(a)$.
 Since
 $|\bigcup_{a\text{ \critical}} \validS(a)|\leq n^2\cdot y$,
 we hence have that
 $|\set{\tau}|\leq n^2\cdot y$
\end{proof}
}

\noindent
To conclude,
\smlaAcr{} admits presumably no problem kernel of size polynomial in~$n$,
but one of size polynomial in~$n+y$.
Interestingly,
for~\fmlaAcr{}
the latter is presumably impossible.

\begin{proposition}[\appref{thm:FMLAnoPKnpy}]
 \label{thm:FMLAnoPKnpy}
 \UnlessPK,
 \fmlaAcr{}
 admits no problem kernel of size polynomial in~$n$,
 even if~$m=2$,
 $k=1$, and~$y=1$.
\end{proposition}

\appendixproof{thm:FMLAnoPKnpy}
{
  \newcommand{\dbin}{\ensuremath{b}}
  Let~$\dbin^i\ceq B_{\tlog{p}}(i)$
  and denote by~$\ol{\dbin_i}$ the complement of~$\dbin^i$.
  Recall that $\dbin^i$ is of length $2\tlog{p}$.
  With $\dbin^i(c)$ we denote that agents corresponding to ones in~$\dbin^i$ nominate candidate~$c$, 
  and nothing otherwise.

  We give an \ORcroco{} from the \NP-hard \prob{X1-3SAT}.
  We can assume that every variable appears in every clause at most once as a literal.

  \begin{construction}
  \label{constr:FMLAnoPKnpy}
    Let $I_1=(X,\phi_1),\dots,I_p=(X,\phi_p)$
    be~$p=2^q$ instances of \prob{X1-3SAT}
    over the same set~$X$ of~$N$ variables,
    where~$\phi_i=\bigwedge_{j=1}^{M} K_j^{(i)}$ for every $i\in\set{p}$
    (we can assume these properties to hold for the input instances 
    as they form a polynomial equivalence relation).
    We construct an instance~$I=(A,C,U,k,x,y)$ with~$x=0$ and
    $k=y=1$ as follows
    (see \cref{fig:FMLAnoPKnpy} for an illustration).
    \begin{figure*}[t!]
    \centering
    \begin{tikzpicture}
      \def\xr{1.1}
      \def\yr{1}
      \def\xs{1.2}
      \def\ys{0.9}
      \tikzpreamble{}
      
      \newcounter{ccount}
      \setcounter{ccount}{0}

      \newcommandx{\mknode}[6][1=1,3=1,5=lightgray]{%
        \pgfmathsetmacro\xn{int(\value{ccount})}
        \node (Ax\xn) at (#2*\xr*\xs,#4*\yr*\ys)[anchor=north west,minimum width=#1*\xs*\xr cm,minimum height=#3*\ys*\yr cm,font=\tiny,draw=#5]{#6};
        \addtocounter{ccount}{1}
      }
      \newcommandx{\mknodeS}[6][1=1,3=1,5=lightgray]{%
        \pgfmathsetmacro\xn{int(\value{ccount})}
        \node (Ax\xn) at (#2*\xr*\xs,#4*\yr*\ys)[anchor=north west,minimum width=#1*\xs*\xr cm,minimum height=0.5*#3*\ys*\yr cm,font=\tiny,draw=#5]{#6};
        \addtocounter{ccount}{1}
      }
      
      \mknode[10]{0}[6]{4}{$\emptyset$};
      \mknode{0}{-2}[red]{$\ol{\dbin^1}(c^*)$}
      \mknode{0}{-3}[red]{$\ol{\dbin^1}(c^*)$}
      \mknode{0}{-4}[red]{$\vdots$}
      \mknode{0}{-5}[red]{$\ol{\dbin^1}(c^*)$}
      \mknode{0}{-6}[red]{$\ol{\dbin^1}(c^*)$}
      \mknodeS{0}{-7}[red]{$c^*$}
      
      \draw [decorate,decoration={brace,amplitude=4pt},xshift=-2pt,yshift=0pt] (0,-3*\yr*\ys) -- (0,-2*\yr*\ys) node [black,midway,xshift=-5pt,anchor=east] {\footnotesize$x_1$};
      \draw [decorate,decoration={brace,amplitude=4pt},xshift=-2pt,yshift=0pt] (0,-4*\yr*\ys) -- (0,-3*\yr*\ys) node [black,midway,xshift=-5pt,anchor=east] {\footnotesize$\lneg{x_1}$};
    
      \draw [decorate,decoration={brace,amplitude=4pt},xshift=-2pt,yshift=0pt] (0,-6*\yr*\ys) -- (0,-5*\yr*\ys) node [black,midway,xshift=-5pt,anchor=east] {\footnotesize$x_N$};
      \draw [decorate,decoration={brace,amplitude=4pt},xshift=-2pt,yshift=0pt] (0,-7*\yr*\ys) -- (0,-6*\yr*\ys) node [black,midway,xshift=-5pt,anchor=east] {\footnotesize$\lneg{x_N}$};
    
      \node at (Ax6.west)[anchor=east]{\footnotesize$a^*$:};
      
      \mknode{1}[5]{-2}[red]{$\cdots$};
      \mknodeS{1}{-7}[red]{$\cdots$};
      
      \mknode{2}{-2}[red]{$\ol{\dbin^p}(c^*)$}
      \mknode{2}{-3}[red]{$\ol{\dbin^p}(c^*)$}
      \mknode{2}{-4}[red]{$\vdots$};
      \mknode{2}{-5}[red]{$\ol{\dbin^p}(c^*)$};
      \mknode{2}{-6}[red]{$\ol{\dbin^p}(c^*)$};
      \mknodeS{2}{-7}[red]{$c^*$};
      
      \mknodeS[9]{3}{-7}{$\emptyset$};
      
      \mknode{3}{-2}[blue]{$\dbin^1(c_\true)$}
      \mknode{3}{-3}[blue]{$\dbin^1(c_\false)$};
      \mknode{3}[3]{-4}[blue]{$\emptyset$};
      
      \mknode{4}{-2}[blue]{$\cdots$}
      \mknode{4}{-3}[blue]{$\cdots$};
      \mknode{4}[3]{-4}[blue]{$\emptyset$};
      
      \mknode{5}{-2}[blue]{$\dbin^p(c_\true)$}
      \mknode{5}{-3}[blue]{$\dbin^p(c_\false)$};
      \mknode{5}[3]{-4}[blue]{$\emptyset$};
      
      \mknode{6}[5]{-2}[blue]{$\cdots$}
      
      \mknode{7}[3]{-2}[blue]{$\emptyset$};
      \mknode{7}{-5}[blue]{$\dbin^1(c_\true)$}
      \mknode{7}{-6}[blue]{$\dbin^1(c_\false)$};
      
      \mknode{8}[3]{-2}[blue]{$\emptyset$};
      \mknode{8}{-5}[blue]{$\cdots$};
      \mknode{8}{-6}[blue]{$\cdots$};
      
      \mknode{9}[3]{-2}[blue]{$\emptyset$};
      \mknode{9}{-5}[blue]{$\dbin^p(c_\true)$}
      \mknode{9}{-6}[blue]{$\dbin^p(c_\false)$};
      
      \mknode{10}{4}{\scalebox{0.75}{$\vdots$}}
      \mknodeS{10}{3}{$\emptyset$}
      \mknodeS{10}{2.5}{$c_\true$}
      \mknodeS{10}{2}{$\emptyset$}
      \mknode{10}{1.5}{\scalebox{0.75}{$\vdots$}}
      \mknodeS{10}{0.5}{$\emptyset$}
      \mknodeS{10}{-0}{$c_\false$}
      \mknodeS{10}{-0.5}{$\emptyset$}
      \mknode{10}{-1}{\scalebox{0.75}{$\vdots$}}
      \mknode{10}{-2}{$\dbin^1(c_\true)$}
      \mknode{10}{-3}{$\dbin^1(c_\false)$};
      \mknode{10}[3]{-4}{$\emptyset$};
      
      \draw [decorate,decoration={brace,amplitude=4pt},xshift=0pt,yshift=2pt] (10*\xs*\xr,4*\yr*\ys) -- (11*\xs*\xr,4*\yr*\ys) node [black,midway,yshift=10pt] {\footnotesize$x_1$};

      \node at (Ax37.west)[anchor=east]{\footnotesize$(x_1 \lor \cdots)=K_i^{(1)} \cong a_i$:};
      \node at (Ax41.west)[anchor=east]{\footnotesize$(\lneg{x_1} \lor \cdots)=K_j^{(1)} \cong a_j$:};
      
      \mknode{11}[6]{4}{$\cdots$};
      \mknode{11}[5]{-2}{$\cdots$};
      
      \foreach \x/\y in {
          0/$L_1^*$, 1/$\cdots$, 2/$L_p^*$, 3/$L_{1,1}$, 4/$\cdots$, 5/$L_{1,p}$, 
          6/$\cdots$, 7/$L_{N,1}$, 8/$\cdots$, 9/$L_{N,p}$, 10/$L_1^{(1)}$, 11/$\cdots$
      }{
        \node at (\x*\xs*\xr+0.5*\xs*\xr,-8*\ys*\yr)[]{\y};
      }
    \end{tikzpicture}
    \caption{Illustration to~\cref{constr:FMLAnoPKnpy}. 
    In this illustrative example, 
    $K^{(1)}_i$ and~$K^{(1)}_j$ are the only clauses of~$\phi_1$ containing~$x_1$ as a literal.}
    \label{fig:FMLAnoPKnpy}
    \end{figure*}
    Let~$C\ceq\{c_\true,c_\false\}$ and~$c^*\ceq c_\true$.
    Let~$A=A'\cup A_X \cup \{a^*\}$,
    where~$A'\ceq \{a_1,\dots,a_M\}$,
    and~$A_X\ceq \{a_{x_i}^\ell,a_{\lneg{x_i}}^\ell\mid i\in\set{N},\ell\in\set{q}\}$.
    \begin{compactenum}[(i)]
    \item There are levels~$L_1^*,\dots,L_p^*$
    where~$a^*$ nominates~$c^*$ in each level,
    and in~$L_j^*$ 
    agents~$a_{x_i}^\ell$ and $a_{\lneg{x_i}}^\ell$
    nominate~$c^*$ if there is a 0 in~$B_{q}(j)$ at position~$\ell$,
    and~$\emptyset$ otherwise,
    for every~$i\in\set{N}$.
    All the other agents nominate nothing.
    \item There are levels~$L_{i,j}$ with~$i\in\set{N}$ and~$j\in\set{p}$
    such that agent~$a_{x_i}^\ell$ nominates~$c_\true$ if there is a 1 in~$B_{q}(j)$ at position~$\ell$
    and agent~$a_{\lneg{x_i}}^\ell$ nominates~$c_\false$ if there is a 1 in~$B_{q}(j)$ at position~$\ell$,
    and each nominate~$\emptyset$ otherwise.
    All the other agents nominate nothing.
    \item There are levels~$L_{i}^{(j)}$ with~$i\in\set{N}$ and~$j\in\set{p}$ 
    such that for each~$h\in\set{M}$,
    we have that
    \[ \text{agent~$a_h$ nominates} \begin{cases}
                                    c_\true, & \text{if~$x_i$ appears in~$K_h^{(j)}$},\\
                                    c_\false, & \text{if~$\lneg{x_i}$ appears in~$K_h^{(j)}$},\\
                                    \emptyset,& \text{otherwise},
                                    \end{cases}
    \] 
    agent~$a_{x_i}^\ell$ nominates~$c_\true$ if there is a 1 in~$B_{q}(j)$ at position~$\ell$
    and agent~$a_{\lneg{x_i}}^\ell$ nominates~$c_\false$ if there is a 1 in~$B_{q}(j)$ at position~$\ell$,
    and each nominate~$\emptyset$ otherwise.
    All the other agents nominate no candidate.
    \end{compactenum}
    This finishes the construction.
    \cqed
  \end{construction}

  Due to~$a^*$ having its only nominations in the levels~$L_1^*,\dots,L_p^*$,
  we have the following.

  \begin{observation}
  \label{obs:FMLAnoPKnpy}
  In every solution
  there is exactly one~$j\in\set{p}$ such that
  the committee in~$L_j^*$ contains~$c^*$.
  \end{observation}

  \begin{lemma}
  \label{lem:FMLAnoPKnpy}
  Given any solution
  where~$j\in\set{p}$ is such that
  the committee in~$L_j^*$ contains~$c^*$,
  then for every~$i\in\set{N}$,
  the committee in each of $L_{i,j'}$ and~$L_{i}^{(j')}$ is disjoint from~$\{c_\true,c_\false\}$ if~$j\neq j'$ 
  and contains exactly one candidate from~$\{c_\true,c_\false\}$ if~$j=j'$.
  \end{lemma}

  \begin{proof}
  Suppose firstly that~$j\neq j'$ but $C_{i,j'}\cap \{c_\true,c_\false\}\neq \emptyset$,
  say~$c_\true\in C_{i,j'}$.
  Due to the unique committee~$C_j^*$,
  there is~$a_{x_i}^\ell$ having~$c^*$ nominated in~$C_j^*$
  and~$c_\true$ in~$C_{i,j'}$
  since there is a position~$\ell$ such that~$B_q(j)$ is 0 and~$B_q(j')$ is 1.
  
  Next,
  suppose that there is~$C_{i,j}\cap \{c_\true,c_\false\}= \emptyset$.
  If~$C_i^{(j)}=\{c_\true\}$,
  then there is~$a_{\lneg{x_i}}^\ell$ being never satisfied.
  If~$C_i^{(j)}=\{c_\false\}$,
  then there is~$a_{x_i}^\ell$ being never satisfied.
  If~$C_i^{(j)}=\emptyset$ or~$C_i^{(j)}=\{c^*\}$,
  then we find two such agents.
  \end{proof}

  \begin{proof}[Proof of~\cref{thm:FMLAnoPKnpy}]
  \RD{} 
  Let~$I_t$ be a \yes-instance.
  Let~$f\colon X\to\{\true,\false\}$ be an exact-1-in-3 truth assignment,
  and let~$\lneg{f}$ be the complementary assignment of~$f$.
  In~$L_r^*$, $c^*$ is nominated if~$r=t$,
  and no candidate is nominated otherwise.
  For each~$i\in\set{N}$,
  in~$L_{i,r}$, $c_{\lneg{f}(x_i)}$ is nominated if~$r=t$,
  and no candidate is nominated otherwise.
  Finally,
  for each~$i\in\set{N}$,
  in~$L_i^{(t)}$ candidate~$c_{f(x_i)}$ is nominated.
  By construction,
  each agent in~$A_X\cup\{a^*\}$ has exactly one committee respecting their nomination.
  Consider an agent~$a_j$ with~$j\in\set{M}$
  having either no committee or at least two committees respecting their nomination.
  Note that by construction,
  we have that~$f(x_i)=\true$ if and only if~$c_\true\in C_i^{(t)}$.
  Hence,
  in either case,
  $f$ is not an exact-1-in-3 truth assignment,
  a contradiction.
  
  \LD{}
  Let~$\calC$ be a solution to~$I$,
  and let~$t\in\set{p}$ such that~$L_t^*$ is the unique level where~$c^*$ is nominated
  (see~\cref{obs:FMLAnoPKnpy}).
  Due to~\cref{lem:FMLAnoPKnpy},
  we know that each of~$C_i^{(t)}\cap \{c_\true,c_\false\}\neq \emptyset$.
  Let~$f\colon X\to \{\true,\false\}$ be such that~$f(x_i)=\true \iff c_\true\in C_i^{(t)}$.
  Since we have that~$a_j$ has a score of~$y$ if and only if 
  there are~$y$ variables in~$K_j^{(t)}$ set to true by~$f$,
  it follows that~$f$ is an exact-1-in-3 truth assignment for~$I_t=(X,\phi_t)$.
  \end{proof}
}

\section{Epilogue}
\label{sec:epilogue}
\appendixsection{sec:epilogue}

We settled the parameterized complexity
for both \smlaAcr{} and \fmlaAcr{}
for several natural parameters and their combinations.
We found that both problems become tractable only
if either the number of agents
or the solution size is lower bounding the parameter.
Hence,
short trips with few per-day activities like in our introductory example
can be tractable even if many agents participate and if there are many activities available.
Also the practically relevant setting where few agents have to select from many options,
where egalitarian or even equitable solutions appear particularly relevant,
can be solved efficiently.

Our two problems have a very similar complexity fingerprint,
yet, they distinguish through the lens of efficient and effective data reduction:
While
\smlaAcr{} admits a problem kernel of size polynomial in~$n+y$,
\fmlaAcr{} presumably does not.
In other words,
it appears unlikely that we can efficiently and effectively
shrink the number of levels for \fmlaAcr{}.

\paragraph{Other Variants.}
Looking at the constraints in \smlaAcr{} and \fmlaAcr{},
one quickly arrives at the following general problem.
Herein,
we generalize to preference functions, where each agent assigns some utility value to each candidate.
Moreover,
we use generalized OWA-based aggregation, 
e.g.,
allowing $\max(\cdot)$ and thus modeling rules such as Chamberlin-Courant.
Let $\sim\,\in\{\leq,=,\geq\}$,
$\Lambda=\{\Lambda_k \in \R^k \mid k \in \N\}$
be a family of (OWA)~vectors, and
${\cal U}$ be a class of preference functions.
We write $\vec{u} (C')$ for the vector of utilities that $u$~assigns to the candidates from~$C'$
sorted in nonincreasing order.
See \citet{BFKKN20} for details.

\newcommand{\vecprod}[2]{\left\langle#1,#2\right\rangle}
\decprob{\bicmce[$\sim_{\rm{k}}$][$\sim_{\rm{x}}$][$\sim_{\rm{y}}$]{$\Lambda,{\cal U}$}}{bicmcegen}
{A set~$A$ of agents,
a set~$C$ of candidates,
a sequential profile of preference functions~$U=(u_{a,t}: C \rightarrow \Nzero \mid a \in A, t\in\set{\tau})$ each from~$\cal U$,
and three integers~\mbox{$k,x,y\in\Nzero$}.}
{Is there a sequence~$C_1,\dots,C_\tau\subseteq C$
such that
\begin{align}
 \text{$\forall t\in\set{\tau}$:} && |C_t| &\sim_{\rm{k}} k, \label{prob:bicmcegen:k}\\
 \text{$\forall t\in\set{\tau}$:} && \sum\nolimits_{a \in A} \vecprod{\Lambda_{|C_t|}}{\vec{u}_{a,t}(C_t)} &\sim_{\rm{x}} x, \label{prob:bicmcegen:x}\\
 \text{and $\forall a\in A$:} && \sum\nolimits_{t=1}^\tau \vecprod{\Lambda_{|C_t|}}{\vec{u}_{a,t}(C_t)} &\sim_{\rm{y}} y?\label{prob:bicmcegen:y}
\end{align}
}

\newcommand{\ourOWA}{SUM}
\newcommand{\ourPREF}{NOM}
\noindent
Let SUM denote the family of OWA-vectors containing only 1-entries,
and NOM be the class of preference functions that contain only 0-entries except for at most one 1-entry.
We have that \smlaAcr{} is \bicmce[$\leq$][$\geq$][$\geq$]{\ourOWA,\ourPREF}
and \fmlaAcr{} is \bicmce[$\leq$][$\geq$][$=$]{\ourOWA,\ourPREF}.
It turns out that all variants except for 
\bicmce[$\leq$][$\leq$][$\leq$]{\ourOWA,\ourPREF}
and
\bicmce[$\geq$][$\geq$][$\geq$]{\ourOWA,\ourPREF} 
are \NP-hard.
In fact,
most of the variants 
(including \smlaAcr{} and \fmlaAcr{})
are \NP-hard even if every voter does not change their vote over the levels.
We defer the details to the appendix~\apprefX{app:sec:epilogue}.

\toappendix
{
  \begin{table*}[t!]
    \centering
    \newcommand{\colBoxA}{magenta!30!white}
    \newcommand{\colBoxB}{orange!40!white}
    \newcommand{\colBoxC}{cyan!70!white}
    \begin{tabular}{r||p{1.25cm}p{1.25cm}p{1.25cm}|p{1.25cm}p{1.25cm}p{1.25cm}|p{1.25cm}p{1.25cm}p{1.25cm}}\toprule
      & \multicolumn{3}{c|}{``$\leq k$''} & \multicolumn{3}{c|}{``$= k$''} & \multicolumn{3}{c}{``$\geq k$''}
      \\\midrule
      $x\backslash y$ & $\leq$ & $=$ & $\geq$ & $\leq$ & $=$ & $\geq$ & $\leq$ & $=$ & $\geq$ \\\midrule\midrule
      $\leq$ & \classP~\tref{obs:easygens} & \cellcolor{\colBoxA}\NP-h.  & \cellcolor{\colBoxA}\NP-h.   
        & \cellcolor{\colBoxB}\NP-h. & \cellcolor{\colBoxA}\NP-h.   & \cellcolor{\colBoxA}\NP-h.  
        & \cellcolor{\colBoxB}\NP-h.   & \cellcolor{\colBoxA}\NP-h.   & \cellcolor{\colBoxA}\NP-h.   \\
      $=$ & \cellcolor{\colBoxB}\NP-h.   & \cellcolor{\colBoxA}\NP-h.   & \cellcolor{\colBoxA}\NP-h.   
        & \cellcolor{\colBoxB}\NP-h.   & \cellcolor{\colBoxA}\NP-h.   & \cellcolor{\colBoxA}\NP-h.  
        & \cellcolor{\colBoxB}\NP-h.   & \cellcolor{\colBoxA}\NP-h.   & \cellcolor{\colBoxA}\NP-h.   \\
      $\geq$ & \cellcolor{\colBoxB}\NP-h.   & \cellcolor{gray!15}\NP-h. \tref{thm:threelevels} & \cellcolor{gray!25}\NP-h. \tref{thm:twolevels}     
        & \cellcolor{\colBoxB}\NP-h.   & \cellcolor{\colBoxA}\NP-h. & \NP-h. \tref{thm:twolevels}
        & \cellcolor{\colBoxB}\NP-h.   & \cellcolor{\colBoxB}\NP-h. & \classP~\tref{obs:easygens} \\
      \bottomrule
    \end{tabular}
    \caption{$\sim_{\rm{x}}$ versus $\sim_{\rm{y}}$ versus~$\sim_{\rm{k}}$
    for \bicmce[$\sim_{\rm{k}}$][$\sim_{\rm{x}}$][$\sim_{\rm{y}}$]{\ourOWA,\ourPREF}. 
    The darkgray-colored cell corresponds to \smlaAcr{},
    and the lightgray-colored cell to \fmlaAcr{}.
    \colorbox{\colBoxA}{\Cref{cor:3part}\phantom{p}}
    \colorbox{\colBoxB}{\Cref{prop:rmis}}
    }
    \label{tab:variants}
  \end{table*}
  \cref{tab:variants} gives an overview on the computational complexity results
  for the variants.
  The following is immediate
  (by selecting no or all candidates).

  \begin{observation}\label{obs:easygens}
    Each of 
    \bicmce[$\leq$][$\leq$][$\leq$]{\ourOWA,\ourPREF}
    and
    \bicmce[$\geq$][$\geq$][$\geq$]{\ourOWA,\ourPREF} 
    is in~\classP{}.
  \end{observation}

  \begin{proposition}\label{prop:rmis}
    Each of the following problems
  is~\NP-hard,
  where $\sim^* \in\{\leq,=,\geq\}$
  and $\sim^\dagger \in\{=,\geq\}$:
  \begin{inparaenum}[(i)]
    \item \bicmce[$\sim^*$][$\geq$][$\leq$]{\ourOWA,\ourPREF},
    \item \bicmce[$\sim^*$][$=$][$\leq$]{\ourOWA,\ourPREF},
    \item \bicmce[$\sim^\dagger$][$\leq$][$\leq$]{\ourOWA,\ourPREF}, \label{prop:rmis:dagger}
    and, 
    \item \bicmce[$\geq$][$\geq$][$=$]{\ourOWA,\ourPREF}.
  \end{inparaenum} 
  \end{proposition}

  The reduction is from the following \NP-hard problem.

  \decprob{Regular Multicolored Indepedent Set (RMIS)}{rmis}
  {An undirected $d$-regular graph~$G=(V,E)$ with~$V=V_1\uplus\dots\uplus V_r$ where each~$V_i$ forms an independent set.}
  {Is there an independent set~$S$ with~$|S\cap V_i|=1$ for all~$i\in\set{r}$?}

  \begin{construction}
  \label{constr:rmis}
  Let~$C\ceq \{c_v\mid v\in V\}\cup C_= \cup\{c_+\}$ where
  $C_= \ceq \{c_=^v\mid v\in V\}\cup \{c_=^e\mid e\in E\}$. 
  Let~$A\ceq A_V\uplus A_E\uplus A_+$ 
  with~$A_V\ceq\{a_v\mid v\in V\}$,
  $A_E\ceq\{a_e\mid e\in E\}$, and
  $A_+\ceq \{a_+^1,\dots,a_+^{d+1}\}$.
  Let~$\tau\ceq r+1$.
  For~$a_v$ with~$v\in V_i$,
  agent~$a_v$ nominates~$c_v$ only in level~$i$,
  $c_=^v$ only in the $(r+1)$st level,
  and nothing in the other layers.
  For~$a_e$ with~$e=\{v,w\}$ and~$v\in V_i$ and~$w\in V_j$,
  agent~$a_e$ nominates~$c_v$ in level~$i$,
  $c_w$ in level~$j$,
  $c_=^e$ in the $(r+1)$st level,
  and nothing in the other layers.
  Each agent from~$A_+$ nominates no candidate in the first~$r$ levels and candidate~$c_+$ in the $(r+1)$st level.
  Let~$k\ceq1$,
  $y\ceq 1$,
  and 
  $x\ceq d+1$.
  \cqed
  \end{construction}

  \begin{proof}
  Follows from the fact that in each variant,
  no agent is allowed to be satisfied more than once.
  However,
  in every level at least one candidate must be chosen,
  either due to the~$x$-constraint or the~$k$-constraint.
  \end{proof}
  
  The \NP-hardness of~\bicmce[$=$][$\geq$][$\geq$]{\ourOWA,\ourPREF}
  follows from the same reduction
  from CBVC
  (see~\cref{ssec:dichotau}) 
  behind~\cref{thm:twolevels}.

  \begin{remark}
    The following problem is in some sense dual to CBVC: 

    \decprob{Constraint Bipartite Indepedent Set (CBIS)}{cbis}
    {An undirected bipartite graph~$G=(V,E)$ with~$V=V_1\uplus V_2$ and~$k_1,k_2\in\N$.}
    {Is there a set~$X\subseteq V$ with~$|X\cap V_i|\geq k_i$ for each~$i\in\{1,2\}$ such that~$G[X]$ contains no edge?}
    
    It is not hard to see that CBIS is \NP-hard via a reduction from CBVC
    by the classic connection of IS to VC.
    Thus, 
    for several problem variants like \bicmce[$\geq$][$\geq$][$\leq$]{\ourOWA,\ourPREF}
    we can adapt the reduction behind~\cref{thm:twolevels}
    using~CBIS to obtain
    \NP-hardness already for two levels.
    \rqed
  \end{remark}

  We next prove that our main problems 
  \smlaAcr{} and \fmlaAcr{} are \NP-hard even if every level looks the same.
  The reduction will also apply for several variants
  (see~\cref{tab:variants}).

  \begin{proposition}\label{prop:3part}
    Each of \smlaAcr{} and \fmlaAcr{}
    is~\NP-hard,
    even if each agent nominates one and the same candidate in every level.
  \end{proposition}
  
  The reduction is from the following strongly \NP-hard problem~\cite{GareyJ79}.

  \decprob{3-Partition}{tpart}
  {A multiset~$S=\{s_1,\dots,s_n\}$ of $n=3m$ positive integers with~$\sum_{i=1}^n s_i = T\cdot m$.}
  {Is there a partition~$(S_1,\dots,S_m)$ into triplets such that for every~$i\in\set{m}$ it holds that~$\sum_{s\in S_i} s = T$?}

  \begin{construction}
    \label{constr:3part}
    Let~$I=(S=\{s_1,\dots,s_n\})$ of $n=3m$ positive integers with~$\sum_{i=1}^n s_i = T\cdot m$
    be an instance of~\textsc{3-Partition}.
    We construct an instance~$I'\ceq (A,C,U,k,x,y)$ as follows.
    Let~$C=\{c_1,\dots,c_n\}$ be the set of candidates,
    let $A=A_1\cup\dots\cup A_n$ with~$A_i=\{a_i^1,\dots,a_i^{s{_i}}\}$
    be the sets of agents,
    let each agent from~$A_i$ nominate~$c_i$ in each of the~$m$ levels,
    and let~$x\ceq T$, 
    $k\ceq 3$, and 
    $y\ceq 1$.
    \cqed
  \end{construction}

  \begin{observation}
    \label{obs:3part}
    If~$I'$ is a \yes-instance,
    then every solution~$(C_1,\dots,C_m)$ is a partition of~$C$ into triples
    with for every~$t\in\set{m}$: $\sum_{c\in C_t} |u^{-1}_t(c)| = x$.
  \end{observation}

  \begin{proof}
  As there are $3m$ candidates, each with an one-to-one correspondence to a set~$A_i$,
  and there are~$m$ levels 
  where in each at most~$3$ candidates can be chosen,
  the first statement follows.
  Assume there is~$r\in\set{m}$ with $\sum_{c\in C_r} |u^{-1}_r(c)| > x=T$.
  Then,
  \begin{align*}
    \sum_{i=1}^n s_i &= \sum_{t=1}^m \sum_{c\in C_t} |u^{-1}_t(c)| 
    \\ 
    &\geq (m-1)\cdot T + \sum_{c\in C_r} |u^{-1}_r(c)| > m\cdot T,
  \end{align*}
  a contradiction.
  \end{proof}

  \begin{proof}[Proof of~\cref{prop:3part}]
    Given instance $I=(S=\{s_1,\dots,s_n\})$ of $n=3m$ positive integers with~$\sum_{i=1}^n s_i = T\cdot m$ of~\textsc{3-Partition},
    we obtain instance~$I'\ceq (A,C,U,k,x,y)$ using~\cref{constr:3part}.
    We show that~$I$ is a \yes-instance
    if and only if
    $I'$ is a \yes-instance.
    
    \RD{}
    Let~$(S_1,\dots,S_m)$ be a solution to~$I$ with~$S_t=\{s_{t,1},s_{t,2},s_{t,3}\}$.
    We claim that~$(C_1,\dots,C_n)$ with~$C_t\ceq \{c_{t,1},c_{t,2},c_{t,3}\}$ is a solution to~$I'$.
    Since~$|C_t|\leq 3=k$ and~$(S_1,\dots,S_m)$ is a partition,
    it remains to show that the score of each~$C_t$ is at least~$x=T$.
    Note that there are exactly~$s_{t,i}$ many agents nominating~$c_{t,i}$.
    Hence, 
    we have that~$\sum_{c\in C_t} |u^{-1}_t(c)| = \sum_{i=1}^3 |u^{-1}_t(c_{t,i})| = \sum_{i=1}^3 s_{t,i} = \sum_{s\in S_t} s = T$.

    \LD{} 
    Follows from~\cref{obs:3part}.
  \end{proof}
  
  In fact,
  the reduction above works for surprisingly many variants.

  \begin{corollary}\label{cor:3part}
  Each of the following problems with $\sim^*\,\in\{\leq,=,\geq\}$ 
  is~\NP-hard,
  even if each agent nominates the same candidate in every level:
  \begin{inparaenum}[(i)]
    \item \bicmce[$\sim^*$][$=$][$=$]{\ourOWA,\ourPREF}, 
    \item \bicmce[$\sim^*$][$=$][$\geq$]{\ourOWA,\ourPREF},
    \item \bicmce[$\sim^*$][$\leq$][$=$]{\ourOWA,\ourPREF},
    \item \bicmce[$\sim^*$][$\leq$][$\geq$]{\ourOWA,\ourPREF},
    and 
    \item \bicmce[$=$][$\geq$][$=$]{\ourOWA,\ourPREF},
  \end{inparaenum} 
  \end{corollary}
}

\paragraph{Outlook.}
Since our model is novel,
also several future research directions 
come to mind.
A parameterized analysis of the variants of \bicmce[$\sim_{\rm{k}}$][$\sim_{\rm{x}}$][$\sim_{\rm{y}}$]{$\Lambda,{\cal U}$}
next to \smlaAcr{} and \fmlaAcr{} could reveal where these variants differ from each other.
One could consider a global budget instead of a budget for each level,
that is,
variants where~$|\bigcup_{t=1}^\tau C_t|\leq k$ or $\sum_{t=1}^\tau |C_t|\leq k$.
Speaking of variants,
another modification could be where the score of any two agents must not differ by more than some given~$\gamma$.
Finally,
as a concrete question:
does \smlaAcr{} or \fmlaAcr{} admit a problem kernel of size
polynomial in~$k$ if $\tau$ is constant?
(Recall that due to~\cref{thm:nopkmtaux},
we know that there is presumably no problem kernel for \smlaAcr{} of size
polynomial in~$\tau$ if $k$ is constant.)

\section*{Acknowledgements}
We thank IJCAI'23 reviewers for constructive feedback.

\printbibliography

\clearpage
\appendix
\section*{Appendix}
\appendixProofText

\end{document}